\documentclass[a4paper,10pt]{article}

	\usepackage{enumerate,enumitem,hyperref,siunitx}
	\usepackage{bm,amsmath,amssymb,amsthm,physics}
	\usepackage{float,authblk,setspace,lineno}
	\usepackage{cite,graphicx}
	\usepackage[symbol,bottom]{footmisc}
    \usepackage{combelow}
	\usepackage[font=small,labelfont=bf,figurename=Figure,labelsep=period]{caption}
	\usepackage[margin=1in]{geometry}
	\usepackage[capitalize]{cleveref}
    \usepackage{svg}
    \usepackage{siunitx}
    \usepackage{booktabs,xcolor}
    \newtheorem{lemma}{Lemma}
    \newtheorem{proposition}{Proposition}

	\newenvironment{keywords}
    {\footnotesize \textbf{Keywords:} }
    { \vspace{0.2cm} }

\begin{document}

	\title{Adaptive therapy under parametric, structural, and measurement uncertainty}
    


	\author[1,2*]{Alexander P Browning}
    \author[3]{Rebecca M Crossley}
    \author[4]{Ryan J Murphy}
    \author[3]{Helen Byrne}
	\author[5*]{Sara Hamis}
	\affil[1]{School of Mathematics and Statistics, University of Melbourne, Australia}
    \affil[2]{ARC Centre of Excellence for the Mathematical Analysis of Cellular Systems, University of Melbourne, Australia}
	\affil[3]{Mathematical Institute, University of Oxford, Oxford, United Kingdom}
	\affil[4]{School of Mathematical Sciences, Adelaide University, Adelaide, Australia}
	\affil[5]{Department of Information Technology, Uppsala University, Uppsala, Sweden}
	

\date{\today}
\maketitle
\footnotetext[1]{Corresponding authors: apbrowning@unimelb.edu.au and sara.hamis@it.uu.se.}


	\begin{abstract}
		\noindent Adaptive therapy has emerged as a promising treatment strategy that exploits within-tumour competition to delay disease progression. Implementation, however, typically relies on indirect measurements of tumour burden and must account for potentially substantial patient heterogeneity. In this work, we capture patient-to-patient variability, parameter uncertainty, and imperfect biomarker measurements with a mathematical and statistical model that we calibrate to clinical prostate cancer data using a Bayesian inference framework. We use the resulting virtual cohort to demonstrate that, within the simple but now well-established Lotka-Volterra-based model, adaptive therapy robustly improves time-to-progression for the subset of patients that are predicted to eventually progress by the model. To account for other risk factors associated with larger tumour volumes, we introduce a new metric based on the risk of metastasis that demonstrates how adaptive therapy may be disadvantageous when sustained tumour burden is also considered. Given the ubiquity of uncertainty in oncology, we then describe several future modelling directions that also capture uncertainty in the temporal evolution of the underlying tumour or biomarker dynamics. Finally, we demonstrate how model misspecification and non-identifiability can lead to unreliable predictions, especially if uncertainty is inadequately captured.
	\end{abstract}


\section{Introduction}

The emergence of resistance presents a significant challenge to cancer treatment, and is thought to be responsible for the majority of cancer-related deaths \cite{Ingham.2025}. Standard cytotoxic treatments prioritise cell killing through continuous application of a maximum tolerated dose, often with a view toward tumour elimination \cite{Gatenby.2020}. This approach, however, presents a strong selective pressure in favour of intrinsically (or evolved) resistant subpopulations. In such cases, an initially strong response to treatment is often followed by recurrence and progression to a treatment-unresponsive tumour \cite{Gatenby.2020}.

Typical approaches to deal with resistance in cancers include combination therapies (e.g. multidrug therapies or chemoimmunotherapy), and the development of new drugs that target the mechanisms driving resistance \cite{Vasan.2019}. Evolutionary therapies, meanwhile, have been proposed to leverage the competitive dynamics between sensitive and resistant subpopulations to trade off tumour elimination for the long-term suppression of resistance \cite{Gatenby.2009a,Gatenby.2009b}. Adaptive therapy, in which dosing is adjusted or withdrawn based on tumour response, has been shown to extend the time-to-progression (TTP) in animal models of ovarian cancer \cite{Gatenby.2009b}, breast cancer \cite{Enriquez-Navas.2016}, and melanoma \cite{Smalley.2019}. More promisingly, small-scale human trials of intermittent dosing for metastatic prostate cancer, with scheduling determined from peripheral biomarker data, demonstrate improved TTP following delivery of less than 50\% of the cumulative dose administered compared to those on continuous therapy \cite{Zhang.2017}. Adaptive therapy is also being tested in ovarian cancer, with a phase II randomised trial adapting carboplatin dose using tumour response in relapsed, platinum-sensitive high-grade serous or endometrioid disease \cite{Mukherjee2024}. However, less explored is the potential for adaptive therapies to increase other risk factors associated with relatively large tumour volumes. Outcomes aside, the lower cumulative dose received in adaptive therapy is associated with many advantages, including reduced financial costs, reduced toxicity, and potentially an improved quality of life \cite{Goldenberg.1995,Crook.2012,Jaswal.2015}. Perhaps most importantly, adaptive therapy does not require the development of new drugs and can potentially be implemented within the confines of existing treatment protocols: in prostate cancer, for example, intermittent therapy is already a recommended treatment option in some cases \cite{Virgo.2023}. 

From a practical perspective, implementation of adaptive therapy relies on feedback between the tumour size and the dosing schedule. For this reason, clinical trials of adaptive therapy have focused on cancers with a biomarker proxy for tumour size, such as prostate-specific antigen (PSA) in prostate cancer \cite{Zhang.2017,Zhang.2022} and CA125 in ovarian cancer \cite{Mukherjee2024}. The prostate cancer adaptive schedule of Zhang \cite{Zhang.2017}, referred to as ``AT50'' \cite{Gallagher.2024}, aims to maintain a gross tumour volume (GTV) within 50 to 100\% of its initial size; in practice this means adaptively applying or withholding treatment to control the patient's PSA level, used as a measurable proxy for GTV. Peripheral or circulating PSA can be obtained relatively non-invasively through weekly or monthly blood draws, however, is known to be only weakly correlated to GTV \cite{Babaian.1995,Eastham.2003}. In particular, PSA is produced by both cancerous prostate cells and healthy prostate tissue, such that elevated PSA levels may also be associated with, for example, urinary tract infections \cite{Ulleryd.1999}. However, GTV can only be obtained through biopsy or imaging \cite{Matsugasumi.2015}. Given that measurements must be taken frequently enough that dosing can be adjusted in response to changes in tumour growth, adaptive therapy for other cancers has been largely limited to in vitro and in vivo models, for which the necessarily accurate measurements of GTV can be taken \cite{Enriquez-Navas.2016}.

The limited and indirect clinical data available necessitates mathematical analysis with correspondingly simple mathematical models. For example, two-species ordinary differential equation (ODE) models that capture the dynamics of a sensitive and resistant cell population in response to treatment have been used to study adaptive therapy \cite{Zhang.2017,Greene.2019,Strobl.2021}, after being used to study intermittent therapy \cite{Ideta2008,Tanaka2010}. Many of these works also assume a deterministic coupling between the GTV and the PSA \cite{Zhang.2017,Brady-Nicholls.2020,Strobl.2021, Ideta2008}. Additionally, implicit in almost all ODE models is the assumption that the tumour dynamics remain fixed for a given patient: that is, that the model with fixed parameters is sufficient to describe the time-evolution of the tumour dynamics throughout a patient's entire course of treatment. Key questions, such as how measurement uncertainty and model misspecification affect the efficacy of adaptive therapy, have yet to be fully investigated.

Existing theoretical works that study variability in adaptive therapy outcomes \cite{Brady-Nicholls.2020,Strobl.2021} typically parameterise models using the Bruchovsky Phase II prostate cancer clinical trial for patients undergoing intermittent androgen deprivation therapy (ADT) in 2006 \cite{Bruchovsky.2006}. We emphasise that this trial was published several years before the rationale for adaptive therapy was presented by Gatenby in 2009 \cite{Gatenby.2009b}, and well before the first adaptive therapy trials (with results from the first fully randomised trial due in 2027 \cite{Zhang.2022}). Significant patient heterogeneity complicates parameterisation, motivating the generation of so-called ``virtual cohorts'' or ``virtual patients'' \cite{Scibilia.2025,Gallagher.2026}, for which plausible, simulated distributions of patient outcomes can be obtained. Large uncertainties are also expected, but not captured by traditional approaches that apply point-parameter estimates to predict adaptive therapy outcomes from clinical trial data for patients undergoing intermittent therapy according to a predetermined schedule.

In this work, we incorporate biomarker measurement and parameter uncertainty into a mathematical analysis of adaptive therapy. In particular, we ask whether adaptive therapy reliably produces better patient outcomes than standard treatments, and whether the patient-specific parameters needed to make such treatment decisions can be inferred from data. Additionally, we aim to demonstrate how temporal changes in tumour dynamics could be incorporated into parameterised mathematical models using time-varying parameters, with the goal of laying the foundation for future research into what we term ``robust~adaptive~therapy''. 

In \cref{sec:model}, we present the established Lotka-Volterra model of tumour dynamics and calibrate it to patient PSA data to construct a virtual cohort that accounts for both patient heterogeneity and parametric uncertainty. From this cohort, we then calculate the model-predicted probability of eventual relapse, and explore how uncertainty in PSA measurements and the relationship between this biomarker and tumour volume affects the ability of adaptive therapy to prolong TTP (\cref{results_ttp}). To  capture risk factors associated with potentially larger GTVs (e.g., metastasis), we introduce a new metric in \cref{sec:riskfactors} that demonstrates a potential trade-off that should be considered in treatment design. Turning from parametric to structural uncertainty, we next demonstrate future modelling approaches that capture dynamic variability in the tumour and PSA dynamics (\cref{sec:dynamic_var}). Given that many models are necessarily parameterised using data from patients who have not received adaptive therapy, in \cref{sec:transport} we perform an analysis to determine the extent to which reliable predictions can be made between treatments. Finally, we conclude with a prospective discussion about the role of uncertainty in adaptive therapies.

\section{Methods}

	\begin{figure}[b!]
	    \centering
	    \includegraphics[width=\textwidth]{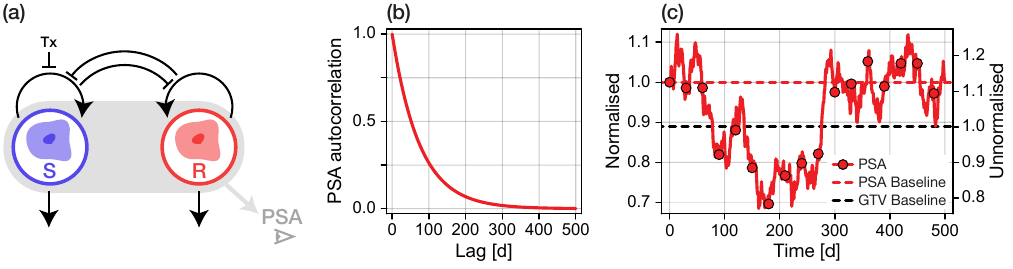}
	    \caption[Figure 1]{\textbf{Mathematical and statistical model of competitive prostate cancer progression.} (a) The Lotka-Volterra model of Strobl et al. \cite{Strobl.2021} assumes a heterogeneous tumour comprised of both drug-sensitive and drug-resistant cells. The proliferation rates of 
        both subpopulations is suppressed 
        due to crowding effects, and in the sensitive cell population due to treatment ($T_\mathrm{x}$). Prostate-specific antigen (PSA) is produced uniformly by all cell types, and forms a proxy measurement for the gross tumour volume (GTV). We assume that PSA measurements are noisy and temporally autocorrelated with a mean proportional to the GTV. We follow \cite{Strobl.2021} and normalise PSA measurements to a \textit{baseline} equal to each patient's initial PSA measurement. (b) \textit{Maximum-a-posteriori} (MAP) autocorrelation function for \textit{Patient 99}. (c) Simulated PSA trajectory for a patient with a constant normalised GTV measurement of unity. We show both a continuous sample (red continuous) and 30 day measurements (red discs). The normalised scale corresponds to a PSA baseline of 1.0, and the unnormalised scale corresponds to a GTV baseline of 1.0. }
	    \label{fig1}
	\end{figure}

    \subsection{Clinical dataset}

    We analyse the publicly available dataset associated with the Bruchovsky Phase II trial \cite{Bruchovsky.2006}. This trial contains data for 109 male prostate cancer patients treated with intermittent androgen suppression. Each patient received the anti-androgen cyproterone acetate as a lead-in therapy for four weeks, followed by a combination of leuprolide acetate and cyproterone acetate for a period of at most 36 weeks. If serum PSA levels were less than \SI{4}{\micro\gram\per\litre} at both 24 and 36 weeks (i.e., at both 168 and 252 days), treatment was withdrawn and PSA measured every four weeks. Treatment was restored when measured PSA had first increased to approximately \SI{10}{\micro\gram\per\litre}. This process was repeated over the course of the clinical trial. Patients whose measured PSA exceeded \SI{4}{\micro\gram\per\litre} at either 24 or 36 weeks were subsequently withdrawn from the study. For each patient, the data comprises the treatment schedule and the measured PSA at approximately four-week intervals. We exclude patients with fewer than five measurements either on or off treatment from our analysis, leaving only 85 patients for whom sufficient data was deemed available. 

	\subsection{Mathematical and statistical model}
	\label{sec:model}
	
	Here, we present a simple mathematical model for the underlying tumour dynamics, alongside a statistical model that describes variability in both measured and simulated PSA measurements. While the underlying mathematical model is deterministic, the statistical model induces stochasticity through an adaptive schedule dependent upon \textit{measured} PSA. In this section, we also define two metrics that we use to compare the performance of continuous and adaptive therapies.
		
	\subsubsection{Tumour dynamics mathematical model}
	
	We consider the Lotka-Volterra model of heterogeneous prostate cancer progression of Strobl et al. \cite{Strobl.2021}, illustrated in \cref{fig1}a. The model assumes that the tumour comprises both drug-sensitive (S) and drug-resistant (R) cells, with volumes denoted by $s(t)$ and $r(t)$, respectively, and with the gross tumour volume (GTV) denoted by $n(t) = s(t) + r(t)$. We assume that both populations occupy the same niche, such that competition is mediated through logistic growth with a fixed, patient-specific, carrying capacity $K$; we revisit this assumption later. In isolation, each cell population proliferates logistically with respective rates $\lambda_S$ and $\lambda_R$, with cell turnover mediated through apoptosis with density-independent rates of $\gamma_T$. Finally, we assume that resistance carries a cost, such that $\lambda_R < \lambda_S$ \cite{Gatenby.2009a} and that conversion of cells between sensitive and resistant types is negligible during treatment \cite{Strobl.2021}.

	Patients in our dataset receive ADT, which is known to target primarily androgen-dependent, drug-sensitive cells \cite{Karantanos.2013}. It is unclear whether ADT primarily inhibits proliferation or, as modelled by \cite{Strobl.2021}, induces death in both mother and daughter cells upon proliferation. We allow both possibilities, through an effective modification on $\lambda_S$ by a multiplicative factor of $1 - \gamma_D T_\mathrm{x}(t)$, where $T_\mathrm{x}(t) \in \{0,1\}$ is an indicator variable representing whether treatment is off or on, respectively, and $0 < \gamma_D < 2$. For $\gamma_D = 0$, the drug has no effect. If ADT inhibits proliferation, then we must have that $0 < \gamma_D < 1$, where $\gamma_D$ represents the factor by which proliferation is reduced, and $\gamma_D = 1$ corresponds to a complete cessation of proliferation in the presence of the drug. If ADT induces cell death upon proliferation, then the factor $0 < \gamma_D / 2 < 1$ represents the probability that a division event results in cell death, such that $0 < \gamma_D < 2$, which unlike the inhibition mechanism, allows for $\gamma_D > 1$. Altogether, these assumptions yield the ordinary differential equation model
	\begin{equation}\label{ode_model}
	\begin{aligned}
		\dv{s(t)}{t} &= \lambda_S s(t) \left(1 - \dfrac{n(t)}{K}\right) (1 - \gamma_D T_\mathrm{x}(t)) - \gamma_T s(t),\\
		\dv{r(t)}{t} &= \lambda_R r(t) \left(1 - \dfrac{n(t)}{K}\right) - \gamma_T r(t),
	\end{aligned}
	\end{equation}
	subject to the initial conditions $s(0) = s_0$ and $r(0) = r_0$.

	Measurements of GTV are not made directly, with adaptive therapy decisions being based on \textit{normalised} measurements of PSA. Additionally, we have no direct measurements of GTV, hence we set the units of GTV such that $K = 1$ for each patient; thus, $s_0 + r_0 \le 1$.

	\subsubsection{Prostate-specific antigen statistical model}
	
	PSA, the concentration of which we denote by $p(t)$, is an indirect biomarker for the GTV, $n(t)$. Two approaches are commonly used in the mathematical literature to relate GTV dynamics to measured PSA; both for model calibration and in the implementation of adaptive therapy. Strobl \cite{Strobl.2021}, and later Gallagher \cite{Gallagher.2024}, do not model PSA explicitly, but rather assume that $p(t)$ is directly proportional to $n(t)$. Brady-Nicholls \cite{Brady-Nicholls.2020}, on the other hand, model the secretion and subsequent degradation of PSA directly such that
		\begin{equation}\label{psa_ode}
			\dv{p(t)}{t} = \underbrace{\alpha n(t)}_{\text{Secretion}} - \underbrace{\beta p(t),}_{\text{Degradation}}
		\end{equation}
    where $\alpha\geq0$ is the secretion rate of PSA by cells and $\beta\geq0$ is the corresponding degradation rate.
	Implicit in both models is the assumption that there is a direct, deterministic relationship between GTV and PSA. While PSA has been observed to correlate with tumour burden \cite{Stamey.1987,Mandel.2021}, this relationship is subject to substantial biological variability and is influenced by factors beyond GTV, not least the volume of benign prostate tissue \cite{Vollmer.2003}. We take an intermediate approach that we revisit later, and assume that measured PSA levels follow a temporally correlated Gaussian process with mean given by the GTV. Such a temporal autocorrelation allows us to capture sustained discrepancy between GTV and PSA, and other possible sources of model misspecification \cite{Lambert.2023}. As such, we have that
		\begin{equation}
			p(t) - n(t) := \varepsilon(t) \sim \mathcal{N}(0,\sigma^2)	
		\end{equation}
	with
		\begin{equation}
			\mathrm{cor}(\varepsilon(t),\varepsilon(t + \delta t)) = \exp\left(-\dfrac{|\delta t| \log(2)}{\tau}\right),
		\end{equation}
	where $\sigma$ is the standard deviation, and $\tau$ denotes the lag at which the correlation drops to 1/2. For simplicity, we choose the natural units of $p(t)$ such that $\mathbb{E}(p(t)) = n(t)$ (i.e., the constant of proportionality is one). This approach is commonly used throughout the adaptive therapy literature, where thresholds are based on a baseline PSA, and analysis uses  a normalised PSA measurement, denoted by $\hat{p}(t)$, by dividing through by the initial concentration
	\begin{equation}
		\hat{p}(t) = p(t) / p_0,
	\end{equation}
	where we denote $p_0 := p(0)$ such that $\hat{p}(0) = 1$. In this way, we do not require the unknown scaling factor that relates units of PSA to GTV.
	
	To see how this model is intermediate to one in which PSA is modelled directly, we extend \cref{psa_ode} to allow for additional, white-noise-driven changes to PSA levels. This yields the It\^{o} stochastic differential equation
	\begin{equation}\label{psa_sde}
		\dd p = \big(\theta n(t) - \beta p(t) \big) \dd t + \eta\, \dd W_t,
	\end{equation}
	where $W_t$ is a Wiener process that satisfies $\mathrm{var}(W_t - W_s) = |t - s|$. If $n(t) \equiv n_\text{const.}$, \cref{psa_sde} is an Ornstein-Uhlenbeck (i.e., Gaussian) process, with equilibrium solution satisfying \cite{Doob.1942}
		\begin{equation}
			\mathbb{E}(p(t)) = \dfrac{\theta}{\beta} \,n_\text{const.},\quad %
			\mathrm{var}(p(t)) = \dfrac{\eta^2}{2\beta},\quad %
			\mathrm{cor}(p(t),p(t+\delta t)) = \exp(-\beta|\delta t|).
		\end{equation}
	That is, our model with scaling factor $\beta / \theta = 1$, $\eta^2 = 2\beta\sigma^2$, and $\beta = \log(2) / \tau$. In principle, all subsequent analysis could be performed directly using \cref{psa_sde}. However, this adds additional statistical considerations, which we explore in \cref{stoch_psa}.
		
	We demonstrate this statistical model in \cref{fig1}b,c, showing a sampled PSA measurement trajectory corresponding to a GTV at $n(t) \equiv 1$, with parameters taken as the maximum-a-posteriori estimate from Patient 99 in the clinical dataset. Measured PSA varies significantly, with a reduction of approximately 20\% sustained over several months, despite a constant GTV.

	\subsubsection{Adaptive therapy}

	As with previous studies, we simulate ADT schedules that depend on measured serum PSA concentration. The goal of adaptive therapy is to maintain the tumour within a fixed size range by applying treatment until the lower set-point is reached, at which point treatment is removed and the tumour allowed to regrow to the higher set-point. In the AT50 schedule implemented clinically by Zhang et al. \cite{Zhang.2022}, these thresholds, denoted by $(\alpha_1,\alpha_2)$, are set to $(0.5, 1)$ on the scale of normalised PSA, $\hat{p}(t)$. Mathematically, we set $T_\mathrm{x}(t)$ such that
		\begin{equation}
			T_\mathrm{x}(t) = \lim_{\delta t \rightarrow 0} \left\{\begin{array}{ll}
					0 & \hat{p}(t) \le \alpha_1,\\
					1 & \hat{p}(t) \ge \alpha_2,\\
					T_\mathrm{x}(t-\delta t) & \text{otherwise}.\\
				\end{array}\right.
		\end{equation}

	In practice, neither precise nor continuous measurements of a patient's GTV are currently feasible. Therefore, we assume that adaptive therapy decisions are based upon noisy PSA measurements collected at regular discrete intervals $t_i = i \Delta$ for $i \in \mathbb{N}$ (i.e., $\Delta = \SI{30}{\day}$ for measurements made approximately monthly). Treatment is, therefore, withdrawn when a normalised PSA measurement is observed to satisfy $\hat{p}(t_i)  \le \alpha_1$, and only resumed when $\hat{p}(t_i) \geq \alpha_2$ is observed. 
	
	\subsubsection{Treatment performance statistics}
	
	\noindent\textit{Time-to-progression.} A patient is deemed to progress if their tumour exceeds a size of $n(t) = \beta n_0$ and is unresponsive to treatment (mathematically, if $n'(t) > 0$ with $T_\mathrm{x}(t) > 0$). We follow previous work \cite{Strobl.2021} and set $\beta = 1.2$, consistent with the established RECIST guidelines, which provide standardised clinical rules to quantify how tumours respond to treatment \cite{Therasse.2000}. While progression would, in practice, be determined using measured PSA since GTV is not measured, our method allows us to determine the actual time until progression using the model-predicted GTV directly.
	
	\medskip
	\noindent\textit{Time-to-metastasis.} Adaptive therapy works by maintaining the tumour at a sufficiently large size that the proliferation of resistant cells that drive progression is suppressed through competition. However, maintaining what may be a large tumour brings other risks, including tumour metastasis. Living with a sustained tumour burden may also be associated with other health risks and reduced quality of life \cite{Chung.2024}. We describe a simple stochastic model of metastasis, whereby every cell at the boundary of the solid tumour (sensitive or resistant) has an equal and instantaneous propensity to metastasise, where metastasis events are driven in individual cells by independent and homogeneous Poisson processes. At the tumour level, therefore, metastasis is driven by an inhomogeneous Poisson process with a rate given by
	\begin{equation}\label{metastasis_rate}
			\lambda(t) = \omega \left(\dfrac{n(t)}{K}\right)^{2/3},	
		\end{equation}
	where $\omega$ determines the relative timescale, and relates to the individual propensity of each cell (we expect $\omega$ to vary between patients, but assume that it is fixed, but unknown, for a given patient). Here, the power of $2/3$ arises from the assumption that the tumour is approximately spherical and that cells are approximately uniform in volume such that the tumour surface area (and hence, number of cells at the tumour periphery) is proportional to $(n(t) / K)^{2/3}$. Finally, by scaling by $K$, we can interpret $\omega$ as the rate of metastasis when the tumour is at carrying capacity, giving us the maximum rate for a patient.
	
	For a given patient, we define the time-to-metastasis (TTM) as the time until the first event in the inhomogeneous Poisson process with rate given by \cref{metastasis_rate}. The TTM can be sampled numerically by discretising \cref{metastasis_rate} as piecewise-constant across sufficiently small intervals $\delta t$. The expected TTM is given by
		\begin{equation}\label{ettm}
			\mathbb{E}(\mathrm{TTM}) = \int_0^\infty \exp\left(-\int_0^t \lambda(u) \dd u \right) \dd t.
		\end{equation}
	To obtain a numerical approximation to this integral, we couple an ODE for $I(t) := \int_0^t \lambda(u) \,\dd u$ to \cref{ode_model} and solve for a sufficiently large time horizon such that the integrand becomes sufficiently small (e.g., the probability that the TTM is larger than the time horizon is negligible). The expectation is then computed using numerical quadrature.
		
	\subsection{Statistical parameter estimation and prediction methods}

	Our primary goal is to determine a distribution of parameter sets that produce realistic patient trajectories, such that we can quantify uncertainty arising from both noisy PSA measurements and incomplete information about the underlying dynamics. A significant challenge when performing inference is that cancer is a highly heterogeneous disease, such that all parameters are likely to vary across the patient cohort. We follow the pseudo-hierarchical approach in our previous work \cite{Browning.2024a}, by constructing a \textit{pooled posterior} as a uniform mixture of the posterior distributions associated with individual patients in our dataset. In this way, we allow all model parameters to vary between patients without placing constraints on the very likely nonlinear correlations between parameters that would be required by a more standard mixed effects or hierarchical modelling approach.
	
	We take a Bayesian approach to inference, and encode initial knowledge about model parameters in a prior, denoted $\pi(\bm\theta)$, where  $\bm\theta$ represents a vector of unknown parameters. The priors for each patient  are independent between parameters, and the marginal priors placed on individual parameters (given in \cref{tab1}) are constructed from the constraints given in existing works. Notably, the parameter $\lambda_S$ is fixed to \SI{0.027}{\per\day} in existing work \cite{Zhang.2017,Strobl.2021}; we relax this by placing a truncated positive normal prior, with untruncated mean of \SI{0.027}{\per\day} and coefficient of variation 0.5. On the parameter $\gamma_D$, which represents the effect of ADT (see \cref{ode_model}), we place a uniform prior over the set of physically constrained values. 
	
	\begin{table}
		\centering
		\begin{tabular}{cp{5cm}cp{4cm}}\hline
			Parameter 	& Description & Value / Prior & Source\\\hline\hline
			$\lambda_S$ 		
				& Sensitive cell proliferation rate 
				& $\mathcal{N}(0.027, 0.0135^2)$\,\SI{}{\per\day}
				& Mean from \cite{Zhang.2017}.\\
			$\lambda_R$	
				& Resistant cell proliferation rate 
				& $c_R \lambda_S$\,\SI{}{\per\day} \\
			$c_R$ 
				& Cost of resistance 
				& $\mathcal{U}(0,1)$
				& Follows from $\lambda_R < \lambda_S$.\\
			$s_0$
				& Initial sensitive cell volume
				& $(1 - f_0) n_0$\\
			$r_0$
				& Initial resistant cell volume
				& $f_0 n_0$\\
			$n_0$
				& Fraction of niche initially occupied
				& $\mathcal{U}(0.1,1)$
				& Lower limit from \cite{Prokopiou.2015,Strobl.2021}\newline
				  (upper limit relaxed)\\
			$f_0$
				& Initial resistant cell fraction
				& $\mathcal{U}(0.001,0.5)$
				& Lower limit from \cite{Grassberger.2019,Strobl.2021}\newline 
				  (upper limit relaxed)\\
			$\gamma_D$
				& Fraction of cells killed by drug at mitosis
				& $\mathcal{U}(0,2)$
				& Physical constraint.\\
			$\gamma_T$
				& Apoptosis rate
				& $\hat{\gamma}_T r_S$\,\SI{}{\per\day} \\
			$\hat{\gamma}_T$ 
				& Relative apoptosis rate
				& $\mathcal{U}(0,1)$
				& Assumption that $\gamma_T < \lambda_S$\\\hline\hline
			$p_0$
				& Initial PSA level
				& $\mathcal{U}(0,2)$\\
			$\sigma$
				& PSA standard deviation
				& $\mathcal{U}(0.001,1)$\\
			$\tau$
				& PSA temporal autocorrelation
				& $\mathrm{Exp}(10)$\,\SI{}{\day}\\
			\hline
		\end{tabular}
		\caption{Model parameter values and priors.}
		\label{tab1}
	\end{table}
	
	Applying data from Patient $i$, we update our knowledge of the parameters to form a corresponding patient-specific posterior distribution, such that
		\begin{equation}
			\underbrace{\pi_i(\bm\theta | \mathcal{D}_i)}_\text{Posterior $i$} \propto \underbrace{\pi(\mathcal{D}_i | \bm\theta)}_\text{Likelihood} \underbrace{\pi(\bm\theta)}_\text{Prior}. 
		\end{equation}
	Here, $\mathcal{D}_i$ represents the PSA measurements and treatment schedule reported for Patient $i$. The likelihood is formed using the correlated normal model described in the previous section.  
	
	Denoting the set of data from all patients as $\{\mathcal{D}_i\}_{i\in\mathcal{I}}$, we can represent the posterior for Patient $i$ as the posterior over all patients, conditioned on the knowledge that the parameters relate to Patient $i$. That is,
	\begin{equation}
		\pi_i(\bm\theta | \mathcal{D}_i) = \pi(\bm\theta | \{\mathcal{D}_i\}_{i\in\mathcal{I}}, i).
	\end{equation}
	In this way, we can marginalise to obtain the posterior over all patients, denoted
	\begin{equation}
		\pi(\bm\theta|\{\mathcal{D}_i\}_{i\in\mathcal{I}}) = \sum_{i \in \mathcal{I}} w_i \pi_i(\bm\theta | \mathcal{D}_i). 
	\end{equation}
	In this work, we take $w_i = w$ for all $i$ to be equal, such that all patients in $\mathcal{I}$ are equally represented. 
	
	Samples from individual patient posteriors are obtained using an adaptive Metropolis with adaptive scaling Markov chain Monte Carlo (MCMC) algorithm provided by \texttt{AdaptiveMCMC} in Julia \cite{Vihola.2020}. For each patient, we take 10\textsuperscript{5} samples across four independent chains, initialised independently using samples from the prior distribution. We verify MCMC convergence using the $\hat{R}$ convergence diagnostic \cite{Vehtari.2021}: in cases for which $\hat{R} < 1.1$ is not satisfied, the chains are discarded and resampled. In some cases, lack of convergence could be attributed to a potential error in the first collected data point: in this case, the first observation is discounted. For 12 of the 85 patients, a fit could not be obtained after 10 MCMC restarts. We attribute this to model misspecification, which we address in \cref{sec:transport-mis} and the discussion; therefore, these patients were excluded from $\mathcal{I}$ (fits for these patients are provided as supplementary material). The concatenation of an equal number of posterior samples from each patient (in our case, $4 \times 10^5$ samples) gives us what is often referred to as a ``virtual cohort'' of parameter values that produce clinically realistic PSA trajectories. We can obtain samples from the posterior over all patients by resampling from this pool.

\section{Results}

	\begin{figure}[!t]
	    \centering
	    \includegraphics[width=\textwidth]{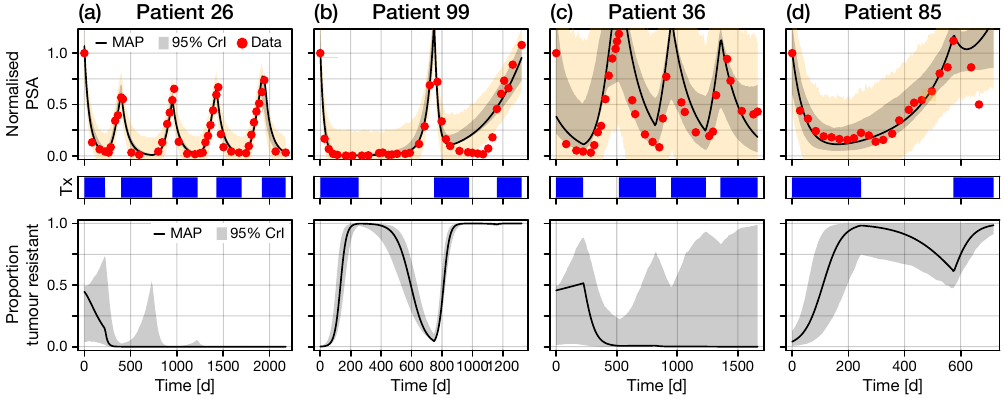}
	    \caption[Figure 2]{\textbf{Individual patient fits and posterior predictions}. To account for both patient-to-patient heterogeneity and uncertainty, we calibrate the mathematical model independently to patients in the Bruchovsky dataset \cite{Bruchovsky.2006} (all patient fits obtained are available as supplementary material). Top row shows normalised PSA measurements (red discs), the model prediction at the MAP parameter set (black), a 95\% credible interval for the normalised PSA (grey), and a 95\% prediction interval for the measured PSA (yellow). It follows from our statistical model that the expected normalised PSA concentration is proportional to the GTV. The reported treatment schedule is given in the middle row indicating periods where the patient was on ADT (blue). Bottom row shows the MAP and 95\% credible interval for the proportion of tumour that is resistant for each patient.} 
	    \label{fig2}
	\end{figure}

	In \cref{fig2} we present the data and individual patient fits for four of the 73 patients in $\mathcal{I}$ (fits for all 73 patients are given as supplementary material). The temporally correlated PSA model also allows us to capture periods in which the PSA deviates from the mean for a sustained time period. For example, the data for Patient 99 (\cref{fig2}b) shows a decrease in PSA that persists following the removal of treatment. This cannot be captured by the tumour dynamics model, which predicts with a high level of certainty that the GTV is increasing. However, the statistical model (demonstrated for Patient 99 in \cref{fig1}c) is able to capture this deviation without necessarily overestimating the measurement noise parameter. Results for Patient 36, in contrast, show greater uncertainty in model predictions, potentially associated with a larger expected deviation between PSA and GTV (e.g., larger noise or a more sustained correlation).

	\subsection{Identifiability of drug action}
    \label{sec:drug_action}

	Examining the posterior for individual parameters (or parameter combinations) allows us to investigate model hypotheses. For example, in \cref{fig3}, we study the posterior for $\gamma_D$ to gain insight into the effect of ADT: whether ADT only inhibits proliferation ($0 < \gamma_D < 1$), or whether it induces cell death upon proliferation ($0 < \gamma_D < 2$). A value $\gamma_D > 1$ is only, therefore, consistent with the death upon proliferation mechanism. These hypotheses cannot be distinguished (are non-identifiable) for $\gamma_D < 1$, although model selection could in principle be performed, with results dependent upon the chosen associated prior over models. In our case, this classification cannot be made for almost all patients (\cref{fig3}a), with 99\% credible intervals for $\gamma_D$ inclusive of the threshold $\gamma_D = 1$. The posterior over all patients (\cref{fig3}b), meanwhile, shows significant mass in the region $\gamma_D > 1$ associated only with the scenario  that the drug induces cell death on proliferation.

	\begin{figure}
	    \centering
	    \includegraphics[width=\textwidth]{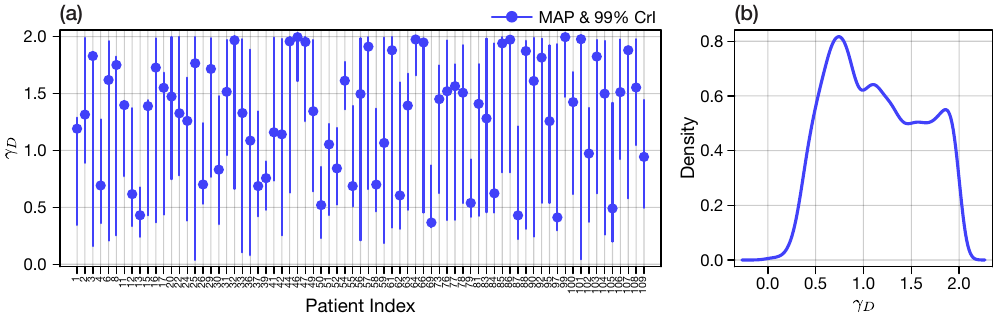}
	    \caption[Figure 3]{\textbf{Posterior summaries for the drug effect, $\gamma_D$}. (a) 99\% credible intervals and MAP produced from the patient-specific posteriors for $\gamma_D$. (b) Joint posterior across all patients for whom a fit was obtained.} 
	    \label{fig3}
	\end{figure}

	\subsection{Predicting the probability of eventual relapse}
    \label{sec:relapse}

	While the mathematical model is calibrated to measurements of PSA, it is the underlying tumour dynamics that are of primary interest. In particular, it is the potential transition from a tumour comprised largely of sensitive cells to one comprised predominantly of resistant cells that leads to eventual relapse. Our Bayesian approach to parameter inference provides information about the quality of the model fit, and also the uncertainty in key statistics such as the proportion of the tumour comprising resistant cells through time (as shown in \cref{fig2}). This uncertainty is particularly important as the parameter values determine the stable equilibria of \cref{ode_model}. Within this model, relapse will only occur if $(s(t),r(t)) = (0,1 - \lambda_R / \gamma_T)$ is the stable equilibrium when $T_\mathrm{x}(t) = 1$; adaptive therapy can potentially \textit{extend} the time-to-progression, but the parameter values determine whether relapse will occur. Full details of the stability analysis are given as supplementary material.
	
	Across all patients in $\mathcal{I}$, the posterior probability of relapse is approximately 43\%, although it is important to note that this is not representative of the general population as our sample is biased towards males enrolled into the trial for whom fits could be obtained. Additionally, this statistic is specific to \cref{ode_model} representing the true underlying dynamics (we expand on this in \cref{sec:transport}). 
	
	The posterior relapse probability varies markedly across the patients shown in \cref{fig2}. Patient 26, for whom the resistant tumour volume tends to zero, has a posterior probability of relapse of approximately 0.7\%. This contrasts with Patients 99 and 85, for whom approximately 100\% of posterior samples correspond to eventual relapse. Most interesting for Patient 99 is the high certainty with which the proportion of the tumour that is resistant varies across the treatment cycle, very quickly tending to a tumour comprised entirely of resistant cells in the first treatment window, before the faster-growing sensitive population dominates towards the end of the period in which the drug was withdrawn. By the second treatment cycle, however, the sensitive cell population appears to be driven extinct, leading to progression. The results for Patient 36 exhibit more uncertainty, with a posterior relapse probability of 35\%. Such an intermediate classification cannot be captured without parameter uncertainty: the best-fit (MAP) model predicts expansion of the resistant population but no relapse, an outcome that is clearly not guaranteed. Crucially, relapse depends on the resistant growth rate $\lambda_R$ (through the equilibrium $1-\lambda_R/\gamma_T$) and, as we show in \cref{sec:transport}, it is this  parameter that adaptive schedules constrain only weakly. This both underlies the uncertainty in the relapse probabilities reported here and limits prediction across schedules.
	
	\subsection{Adaptive therapy robustly improves time-to-progression}
    \label{results_ttp}

	Equipped with a calibrated mathematical model, we turn to prediction in order to compare adaptive therapy with continuous treatment. For any given patient, our approach captures two primary sources of uncertainty: first, as explored in the previous section, parameter uncertainty due to finite, noisy measurements. Second, uncertainty from decisions based on PSA: a noisy and indirect measurement of GTV. Our statistical approach, in particular, captures the possibility that the PSA deviates significantly from GTV for a sustained period of time. Our goal is to determine whether, within our mathematical model, but across all potential parameter sets, adaptive therapy presents an improved TTP.
	
	\subsubsection{Within a single patient}
	
	We focus first on Patient 99, for whom almost all parameter sets yield relapse (e.g., a finite TTP). We consider the so-called ``AT50'' schedule, wherein treatment is applied until \textit{measured} PSA (that is, PSA measured inclusive of the autocorrelated noise model) falls below 50\% of the initial value. Treatment is then removed until the measured PSA exceeds the initial reading (e.g., normalised PSA, $\hat{p}(t) \ge 1$). Additionally, we apply a clinically realistic decision interval of $\Delta = \SI{30}{\day}$. 
	
	\begin{figure}[!t]
	    \centering
	    \includegraphics[width=\textwidth]{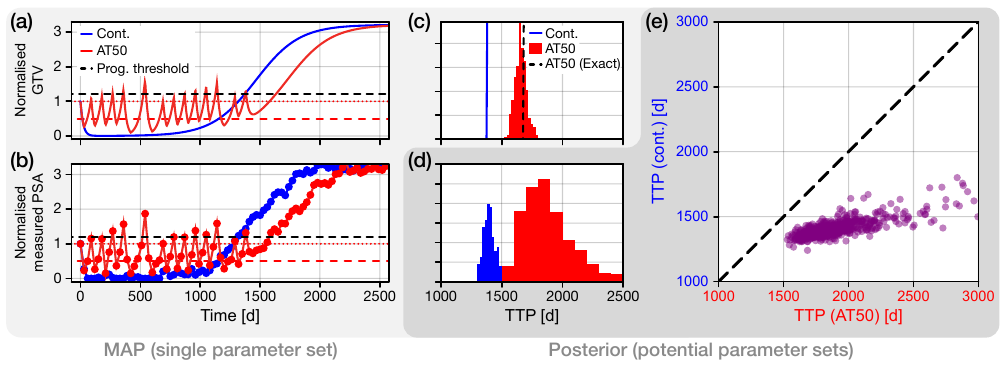}
	    \caption[Figure 4]{\textbf{Simulated adaptive treatment schedules.} We compare continuous treatment and a so-called ``AT50'' adaptive schedule with a 30-day decision interval.  In (a,b,c), we show results for a single parameter set, the MAP for Patient 99; and in (d,e), across the posterior for Patient 99. (a) GTV, normalised such that the initial volume is unity. (b) Normalised measured PSA, showing continuously sampled trajectories and (discs) 30-day measurements at which time treatment decisions are made. In (a,b), the adaptive therapy thresholds (red dashed and dotted for 50\% and 100\%, respectively) are violated due to a finite decision interval and an imperfect correlation between GTV and PSA. (c,d) Time-to-progression (TTP) for both treatment strategies using (c) the MAP and (d) samples from the posterior distribution. In (c), the TTP on continuous treatment is deterministic as stochasticity enters only through noise in PSA measurements. Also shown is the TTP under the assumption that PSA corresponds precisely to GTV (vertical black dashed). (e) Parameters are resampled from the posterior and the TTP under both continuous and adaptive therapy sampled. Only samples with a finite TTP are shown.}
	    \label{fig4}
	\end{figure}
	
	In \cref{fig4}a,b, we produce a single realisation of both continuous therapy and an AT50 schedule using the MAP for Patient 99, showing the normalised GTV, $n(t) / n(0)$, in \cref{fig4}a; and, the normalised measured PSA, in \cref{fig4}b. In contrast to similar studies, which assume precise, continuously-collected measurements, our results yield normalised GTVs that violate the AT50 goal of keeping the tumour within 50--100\% of its initial volume. This is because, here, decisions are based on both noisy measurements (e.g., measurement noise may indicate that the GTV has shrunk further than it actually has), and measurements that are not made continuously (e.g., it may take up to 30 days for a tumour that has breached the threshold to be detected). In \cref{fig4}c, we simulate 500 such trajectories and show the distribution of TTP. We remind the reader that there is no parameter uncertainty here: variability is solely due to a finite decision interval and noisy measurements of PSA, thus we expect no variability under continuous therapy. Clearly, the AT50 schedule results in a higher TTP, even when measurement noise is considered. Furthermore, the median TTP appears similar to that from a schedule where exact measurements of GTV are available: including measurement noise in the decision process appears as likely to raise the TTP as to lower it. 
	
	Next, we repeat the analysis across the posterior for Patient 99, incorporating all sources of variability captured by our framework. In \cref{fig4}d, we show the marginal distribution of TTP, now with variability in TTP under continuous therapy due to parameter uncertainty. In all cases, as we would expect, variability is much higher when the parameters are fixed (in \cref{fig4}c). Finally, in \cref{fig4}e we plot samples from the full distribution: each point corresponding to a fixed parameter set sampled from the posterior for Patient 99, along with the continuous therapy TTP and that sampled from a single AT50 schedule. For this patient, adaptive therapy always outperforms continuous therapy.
	
	\subsubsection{Across all patients}
	
	In Patient 99, for whom our model will almost always predict eventual relapse, we see that adaptive therapy provides an improved TTP relative to continuous therapy. However, it is only for 43\% of parameter values in our posterior over all patients for whom we expect to see relapse, based on a steady-state analysis (supplementary material). Here, we investigate the TTP improvement provided by adaptive therapy by sampling across all patients.
	
	For the AT50 30-day schedule studied in \cref{fig4}, we show the posterior-over-all patients analogue in \cref{fig5}a. For approximately 95\% of samples, we again see that AT50 outperforms continuous therapy in terms of TTP. In cases where continuous therapy outperforms AT50,  the improvement is, for the most part, small. However, only approximately 12\% of samples result in progression within the \SI{e5}{\day} horizon over which we run our simulations. That is, the majority of patients (including those we expect to eventually relapse) are not observed to progress, even under continuous therapy. 
	
	In \cref{fig5}b we demonstrate that this finding is consistent across a range of adaptive schedules, for which the threshold and decision interval vary. In all cases, the adaptive schedule is an improvement on continuous treatment. Further, the highest average improvement in \cref{fig5}b appears to be associated with the AT50 schedule: this is expected, as adaptive therapy exploits the competition dynamics between sensitive and resistant cells, with growth suppression due to this competition being strongest in our model for large GTVs (and hence, larger adaptive therapy minimum thresholds $\alpha_1$).

	\begin{figure}[!t]
	    \centering
	    \includegraphics[width=\textwidth]{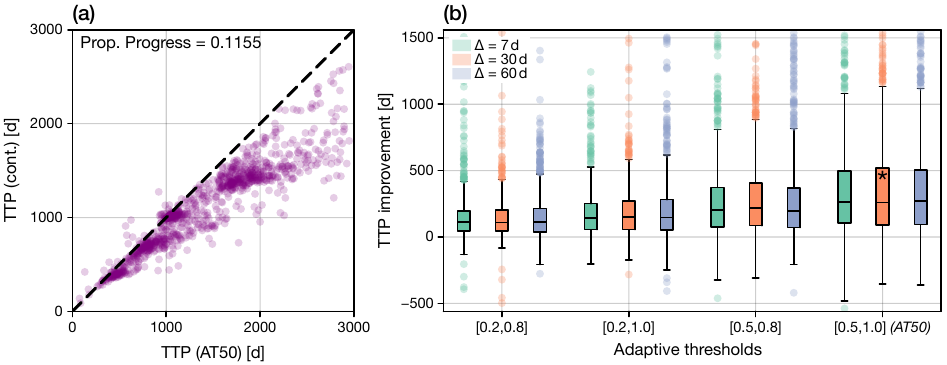}
	    \caption[Figure 5]{\textbf{Time-to-progression improvement across all patients.} (a) We repeat the analysis in \cref{fig4}e for the joint posterior distribution across \textit{all} patients for whom fits were obtained. A finite TTP, indicating that a patient eventually progresses is observed in only 11.55\% of samples. For almost all samples, the TTP is higher (clinically, an improvement) under adaptive therapy. (b) We examine the TTP improvement ($\mathrm{TTP}_\text{adapt.} - \mathrm{TTP}_\text{cont.}$) for a suite of thresholds and decision intervals. An asterisk (*) indicates the box plot corresponding to the 30-day decision AT50 scenario shown in (a) and in \cref{fig4}. }
	    \label{fig5}
	\end{figure}

	\subsection{Adaptive therapy potentially associates with higher risk factors}
    \label{sec:riskfactors}

	Fundamentally, adaptive therapy represents a trade-off between the potential elimination of sensitive cells and the maintenance of a larger (but stable) GTV within which the growth of resistant cells is suppressed by competition (see e.g., \cref{fig4}a). While this approach can lead to more favourable outcomes as measured by TTP, there are various risk factors associated with large tumours. Here, we focus on metastasis, with the TTM modelled by the inhomogeneous GTV-dependent Poisson process given by \cref{metastasis_rate}. We highlight, however, that these results may provide insight into other risk factors correlated with GTV that are not captured by the model; for example, evolution to more aggressive genotypes. 
	
	In \cref{fig6}a, we sample the TTM for Patient 99, for whom we have shown that---within our modelling framework---an AT50 schedule almost certainly yields a TTP improvement. It is only with a posterior probability of approximately 27\% that the AT50 schedule is predicted to yield an improvement in TTM. Results in \cref{fig6}b, meanwhile, show the expected TTM (\cref{ettm}) across the patient-specific posterior for Patient 99; in all cases, the expected TTM is lower under continuous therapy. Next, we repeat the analysis to sample TTM across the posterior-over-all patients (\cref{fig6}c), distinguishing patients that eventually progress under continuous therapy from those that do not. For both groups, continuous therapy once again outperforms the AT50 schedule.
	
	We emphasise that the results in \cref{fig6}a are  illustrative, and depend on the patient-specific instantaneous risk parameter $\omega$, which determines the propensity with which individual cells metastasise. Here, we have set $\omega = \SI{e-2}{\per\day}$ for all patients, although this parameter is likely to vary between patients. For patients with relatively large $\omega$, we expect metastasis to occur relatively early: thus, the immediate reduction in GTV sustained by continuous therapy presents an advantage. For patients with a relatively small $\omega$, the relative risk is lower at all GTVs: adaptive therapy may present an advantage by delaying the TTP and, thus, time until the risk increases due to tumour growth exceeding its initial size. In practice, we expect $\omega$ to be difficult to parameterise a priori, without patient-specific \textit{in vitro} information provided by, for example, organoid studies \cite{Drost.2018}.

	\begin{figure}
	    \centering
	    \includegraphics{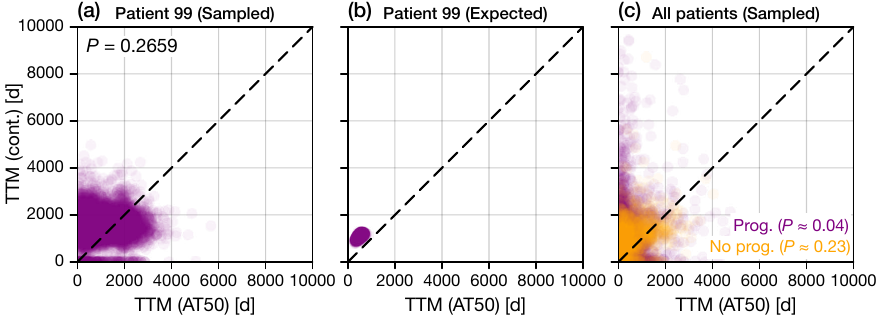}
	    \caption[Figure 6]{\textbf{Time-to-metastasis.} We compare the time-to-metastasis ($TTM$) for adaptive and continuous therapy when metastasis occurs according to an inhomogeneous Poisson process with instantaneous rate $\lambda(t) = \omega (n^{2/3}(t) / K^{2/3})$, where $n(t)$ is the normalised GTV. Here, we fix $\omega = \SI{e-2}{\per\day}$, although it is likely to vary between patients. Results in (a,c) show the sampled time-to-metastasis across resampled posterior samples associated with (a) Patient 99, and (c) the joint posterior for all patients. In all cases, $P$ gives the proportion of samples for which adaptive therapy outperforms continuous therapy, indicated by a longer time-to-metastasis. In (b), we compute the expected time-to-metastasis across posterior samples for Patient 99.}
	    \label{fig6}
	\end{figure}

	\subsection{Incorporating dynamic variability}
    \label{sec:dynamic_var}
	
	Implicit in almost all deterministic models of tumour progression is that the underlying dynamics and parameter values are fixed. Over the treatment timescales considered in this work (the average duration in our dataset is more than four years), it is reasonable to expect significant variability in key model parameters: due both to short-term fluctuations associated with variability in patient health, and to long-term trends driven by tumour evolution.

	\subsubsection{Temporal variation in tumour dynamics}

	First, we consider a scenario where the tumour carrying capacity is dynamic such that $K \mapsto K(t)$. We are motivated by observations of some patients whose data show apparent differences in dynamics between treatment cycles: including, in some cases, reductions in GTV while off treatment (e.g., Patient 2 for $t > \SI{1600}{\day}$ in \cref{fig7}). Such changes in carrying capacity could be driven by one of the many hallmarks of cancer, including tumour remodelling, vascularisation, changes in the immune response, among many others \cite{Hanahan.2011}.
	
	To capture both short-term fluctuations and long-term trends, we consider that $K(t)$ is driven by an Ornstein-Uhlenbeck process with long-term average given by $K_\infty$, to be inferred. This yields
	\begin{equation}
		\dd K(t) = -\xi_K(K - K_\infty) \dd t + \eta_K \dd W_t.
	\end{equation}
	Here, without loss of generality (since the data are rescaled), we set $K(0) = 1$ such that $K_\infty > 1$ corresponds to a larger tumour niche in the long-term, and conversely when $K_\infty < 1$. The parameters $\xi_K$, corresponding to the rate at which $K(t)$ tends to $K_\infty$, and $\eta_K$, which characterises the magnitude of short-term fluctuations, are inferred alongside $K_\infty$.
	
	We perform inference on the new system of stochastic differential equations using a pseudo-marginal particle filter method adapted from \cite{Warne.2020}. Given the substantial computational cost of parameter inference, we restrict our analysis to a single motivating patient, Patient 2, with results (for a single posterior sample) shown in \cref{fig7}. Unique to this model is that it is the underlying tumour dynamics that are stochastic: for a fixed parameter set and treatment regime, we see variability in the expected tumour size (\cref{fig7}). 

	\begin{figure}[!t]
	    \centering
	    \includegraphics{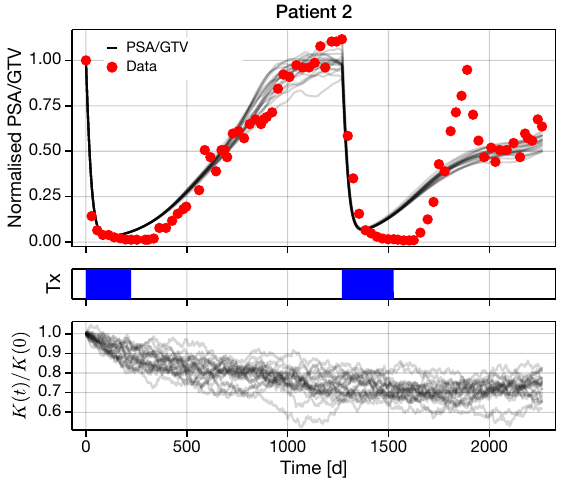}
	    \caption[Figure 7]{\textbf{Stochastic carrying capacity model.} We calibrate the stochastic carrying capacity model to data from Patient 2 using a pseudo-marginal particle filter \cite{Warne.2020}. 
        (Top) We compare clinical data to several realisations of the model for a fixed parameter set sampled from the patient-specific posterior distribution. Under our statistical model, the expected PSA concentration is proportional to GTV, thus both expected PSA and GTV are identical when normalised. The treatment schedule is given in the middle row indicating periods when the ADT was applied (blue). 
        (Bottom) For the fixed parameter set we sample trajectories for $K(t)$, showing how the tumour dynamics vary across the patient's course of treatment. Here, $K_\infty \approx 0.713$.}
	    \label{fig7}
	\end{figure}

	\subsubsection{Temporal variation in PSA dynamics}\label{stoch_psa}
	
	Secondly, we consider a scenario in which a patient's PSA dynamics change with time. This could be, for example, due to fluctuations in the expression of genes linked to PSA production \cite{Thorek.2013}. Specifically, we now model PSA secretion and degradation explicitly, and assume that the secretion rate, denoted by $\alpha(t)$, is dynamic. This yields a set of stochastic differential equations coupled to the deterministic tumour dynamics model considered previously in \cref{ode_model}, with
	\begin{subequations}\label{stoch_psa_model}
	\begin{align}
		\dv{p(t)}{t} &= \alpha(t) n(t) - \beta p(t)\label{stoch_psa_model_a},\\
		\dd \alpha(t) &= -\xi_\alpha (\alpha - \hat{\alpha}) \dd t + \eta_\alpha \dd W_t\label{stoch_psa_model_b}.
	\end{align}
	\end{subequations}
	Here, $\beta$ represents the constant PSA degradation rate, $\hat{\alpha}$ the mean PSA secretion rate, and $\xi_\alpha$ and $\eta_\alpha$ parameters that determine the variance and autocorrelation of $\alpha(t)$. Initially, we assume that the system satisfies $p(0) = p_0$, as before, and that $\alpha(t)$ is stationary, such that $\alpha(t) \sim \mathcal{N}(\hat{\alpha}, \eta_\alpha^2 / (2\xi_\alpha))$. One weakness of this model is that $\alpha(t)$ is not guaranteed to be positive, potentially leading to PSA trajectories that include negative values.
	
	As a linear system, however, \cref{stoch_psa_model} yields a semi-analytical solution in terms of integrals relating to $n(t)$, itself described by the original ODE model (\cref{ode_model}) which is not coupled to the equation of $p(t)$. The mean PSA is given simply by
		\begin{equation}\label{stoch_psa_model_mean}
			\dv{\langle p(t) \rangle}{t} = \hat\alpha n(t) - \beta p(t).
		\end{equation}
	To obtain an expression for the covariance, we multiply \cref{stoch_psa_model_a} by an integrating factor to obtain
		\begin{equation}\label{stoch_psa_model_sol}
			p(t) = \mathrm{e}^{-\beta t} \left(p(0) + \int_0^t \mathrm{e}^{\beta u} \alpha(u)n(u) \dd u\right).
		\end{equation}
	If  $\alpha(t)$ is a Gaussian process, then $p(t)$ is also Gaussian. Therefore, the dynamics of $p(t)$ are entirely described by the mean (given above) and the autocovariance. From the solution to \cref{stoch_psa_model_b}, it follows that 
		\begin{equation*}
			\mathrm{cov}(p(t_1),p(t_2)) = \dfrac{\eta_\alpha}{2\xi_\alpha}\int_0^{t_1}\int_0^{t_2} n(u_1)n(u_2) \mathrm{e}^{-\xi_\alpha|u_1 - u_2|} \dd u_2 \dd u_1.
		\end{equation*}

	In practice, we do not work directly with this integral formulation, but rather transform it into a system of ODEs that can be coupled to both the equation for the mean (\cref{stoch_psa_model_mean}), and the equations that characterise $n(t)$ (\cref{ode_model}). Specifically, it can be shown that
		\begin{equation*}\label{stoch_psa_model_cov}
			\mathrm{cov}(p(t_1),p(t_2)) = \dfrac{\eta_\alpha}{2\xi_\alpha}\Big(A_+(t_1)A_-(t_2) + B_+(t_1) - B_-(t_1)\Big),
		\end{equation*}
	where
		\begin{equation}
			\dv{A_\pm(t)}{t} = n(t) \mathrm{e}^{\pm \xi_\alpha t}\qquad\text{and}\qquad\dv{B_\pm}{t} = A_\pm(t) n(t) \mathrm{e}^{\mp \xi_\alpha t},
		\end{equation}
	subject to $A_\pm(0) = B_\pm(0) = 0$. A full derivation is provided in the supplementary material.

	\begin{figure}[!t]
	    \centering
	    \includegraphics{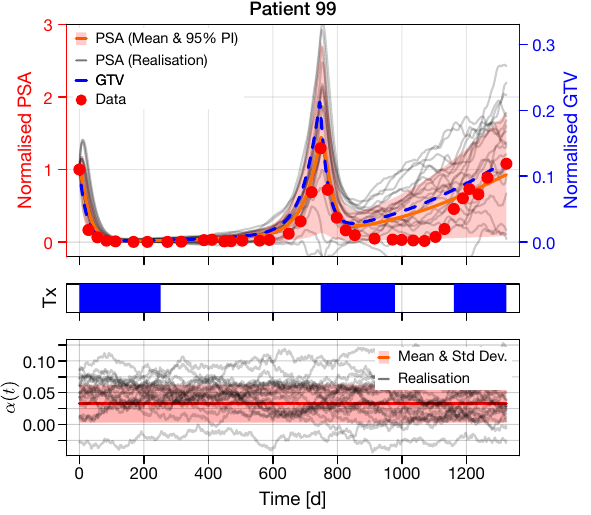}
	    \caption[Figure 8]{\textbf{Stochastic PSA secretion model.} We calibrate the stochastic PSA secretion model to data from Patient 99 by solving the larger coupled system of ODEs that give the Gaussian process solution to the PSA concentration. Here, we show results using a single parameter set given by the MAP. (Top) We compare clinical data to several PSA realisations and the mean PSA (left axis) and the GTV (right axis). Both the PSA and GTV are normalised to appear as unity on the normalised PSA axis. (Bottom) We show several realisations for the secretion rate, along with the mean and standard deviation. }
	    \label{fig8}
	\end{figure}
	
	Solving the resultant coupled ODE system gives sufficient information to parameterise the Gaussian process that determines the distribution solution to $p(t)$, and from which we can formulate a likelihood function. Compared to the stochastic model considered in the previous section, this allows us to perform Bayesian inference with a similar computational cost to the approach applied in the rest of this paper. In \cref{fig8}, we show such results using the MAP parameter set for Patient 99, considered earlier. The most obvious benefit of the stochastic PSA secretion model is that the variability in PSA measurements scales much more naturally: during periods of rapid change (e.g., due to a rapidly growing tumour), there is significant variation in PSA, even though the GTV is deterministic. Additionally, comparing the mean PSA trajectory to that of the GTV shows a delay, whereby growth in tumour size is not immediately reflected in PSA measurements. Finally, examining individual trajectories shows  more clearly how sustained PSA deviations from the GTV can arise.

	\subsection{Transporting inferred model parameters across treatment schedules} \label{sec:transport}

	In this work, we have constructed a cohort of what are often referred to as ``virtual patients'' by calibrating a mathematical model to data from patients that have undergone a fixed treatment schedule. The virtual cohort is then used for prediction, in much the same way as what is often referred to as an ``\textit{in silico} clinical trial'' or a ``digital twin''. Typically, these methods assume that the model structure and inferred parameter values under the observed treatment \textit{transport} (that is, are applicable) to simulated alternative treatments. That is, we have assumed that both the model and its parameters are treatment-schedule-invariant. In this section, we isolate two mechanisms by which parameter transport may fail between no-drug, continuous, intermittent, and adaptive schedules: due to non-identifiability, or due to model misspecification.
	
	\newcommand{\schobs}{{\mathcal{S}_\mathrm{obs}}}
	\newcommand{\schnew}{{\mathcal{S}_\mathrm{new}}}
	
	\subsubsection{Cross-schedule prediction}

	For any treatment schedule, denoted by $\mathcal{S}$ and defined by its input $T_\mathrm{x}(\cdot)$, let $G_\mathcal{S}(\bm\theta;t)$ denote the modelled GTV $n(t)$ obtained from \cref{ode_model}. Suppose that there exists a hypothetical true parameter set, denoted by $\bm\theta^\ast$. Given data from source schedule $\schobs$, the posterior $\pi(\bm\theta\mid y_\schobs)$ has MAP estimate $\hat{\bm{\theta}}_\schobs$. We introduce the \textit{transport error}, denoted by $\Delta_{\schobs\to \schnew}$, as the error incurred when predicting the GTV under target schedule $\schnew$ using the source-inferred parameters $\hat{\bm{\theta}}_{\schobs}$ in place of the true $\bm{\theta}^\ast$,
    \begin{equation}\label{eq:transporterror}
	   \Delta_{\schobs\to \schnew}(t) := 
       G_\schnew(\hat{\bm{\theta}}_{\schobs};t) - n_\schnew^\ast(t),
\end{equation}
	where $n_\schnew^\ast(t)$ is the hypothetical true, noise-free GTV under schedule $\schnew$, which coincides with $G_\schnew(\bm{\theta}^\ast;t)$ only when the model is correctly specified. 
    Whenever the model is correctly specified and the fit recovers the truth, $\hat{\bm{\theta}}_{\schobs}=\bm\theta^\ast$, we have $\Delta_{{\schobs}\to \schnew}\equiv 0$. 
    Therefore, a non-zero transport error means that the best-fit parameters for $\schobs$ do not transport to $\schnew$, for at least one of two reasons: either the source data do not constrain the parameters on which the target depends (non-identifiability), so the best-fit estimate cannot be trusted for unobserved schedules even though the full posterior may still support the truth; or no single $\bm\theta$ fits all schedules (misspecification), so the error is systematic and remains even for the full posterior.

	\subsubsection{Non-identifiability}
	\label{sec:transport-ide}
    In clinically relevant situations, parameters such as the initial resistant fraction $f_0$ are not directly measurable. We therefore study parameter identifiability theoretically, across  relevant treatment schedules, rather than drawing conclusions from our specific patient cohort.
    Let $J_\schobs=\partial G_\schobs/\partial\bm\theta\big|_{\bm\theta^\ast}$ be the sensitivity matrix of the GTV under the source schedule, with rows indexed by the measurement times and columns by the parameters.
    Its nullspace $\mathcal{N}_\schobs = \ker J_\schobs$ is the set of locally structurally non-identifiable directions \cite{Rothenberg1971}.
    These are the parameter perturbations $\delta\bm\theta\in\mathcal{N}_\schobs$ along which $G_\schobs$ is unchanged to first order, i.e., $J_\schobs\,\delta\bm\theta = 0$.
    Now, let $g_\schnew(\bm\theta)$ be any scalar target quantity under schedule $\schnew$ (e.g., time-to-progression, or the GTV at a specified time).

	\begin{lemma}[Transportability condition]
	\label{lem:transport}
	To leading order, noise-free data from the source schedule $\schobs$ constrain $\bm\theta$ only up to $\mathcal{N}_\schobs$. 
    The target $g_\schnew$ is uniquely determined by the data $\schobs$ 
    iff $\nabla g_\schnew(\bm\theta^\ast)\perp \mathcal{N}_\schobs$; 
    otherwise predicting $g_\schnew$ incurs an error $\nabla g_\schnew(\bm\theta^\ast)^\intercal\delta\bm\theta$ for some displacement $\delta\bm\theta\in \mathcal{N}_\schobs$ that the source data do not determine. 
	\end{lemma}
	\noindent This is the classical local-identifiability condition \cite{Rothenberg1971} applied to a scalar target; 
    see supplementary material for details and proof.

	In practice, we may have that model parameters are structurally, but not always practically, identifiable \cite{Raue2009}. 
    We therefore relax our argument from the exact nullspace $\mathcal{N}_\schobs$ to the \textit{sloppy} directions \cite{Gutenkunst2007}: the eigenvectors of $F_\schobs := J_\schobs^\intercal J_\schobs$ with relatively small (but non-zero) eigenvalues.
	%
	%
	The Fisher information $F_\mathcal{S}$ quantifies how much the data from a schedule $\mathcal{S}$ constrain the parameters; along a direction $\mathbf{v}$, the scalar $\mathbf{v}^\intercal F_\mathcal{S}\,\mathbf{v}$ is large when the data tightly determine the parameter combination $\mathbf{v}$ and small when they leave it poorly constrained.
    %
    
    We define the \textit{transportability index} $\rho_{\schobs\to \schnew}$ as the ratio of the target-$\schnew$ to source-$\schobs$ Fisher information along $\mathbf{v}$,
    \begin{equation}\label{eq:rho}
    \rho_{\schobs\to \schnew}(\mathbf{v})=\frac{\mathbf{v}^\intercal F_\schnew\,\mathbf{v}}{\mathbf{v}^\intercal F_\schobs\,\mathbf{v}}.
    \end{equation}
	Then $ \rho_{\schobs\to \schnew}(\mathbf{v})\gg 1$ along directions $\mathbf{v}$ that are poorly constrained by the source $\schobs$ and on which the target $\schnew$ depends. There, the target leans on information not provided by the source, 
    so transport along $\mathbf{v}$ is unreliable. As such, point estimates drawn from a source schedule fail to transport to any target whose prediction depends on a direction it leaves unconstrained, whether exactly (\cref{lem:transport}) or in practice (large but finite $\rho_{\schobs\to \schnew}(\mathbf{v})$). 

    The directions that source data leaves unconstrained depend on the schedule for which it was obtained.
    A time-invariant input (continuous or no-drug) probes only one drug-state and so leaves whole directions unconstrained.
    A time-variant input, whether intermittent (fixed on/off drug cycles) or adaptive (dosing modulated by estimated tumour burden), excites both drug-states and thus constrains the model parameters more. 
    In a full Bayesian treatment, this manifests as posterior uncertainty: 
    with an appropriate choice of prior the posterior spreads along the unconstrained directions yet still covers the truth. 
    Under non-identifiability, a point estimate such as the MAP $\hat{\bm{\theta}}_\schobs$ is therefore unreliable 
    (an arbitrary choice from a wide but honest posterior) rather than systematically wrong. 
    Misspecification, by contrast, biases even the full posterior (\cref{sec:transport-mis}).
    
	\medskip
	\noindent\textit{Time-invariant source schedules.} 
    Under continuous therapy ($T_\mathrm{x}\equiv1$), 
    $\lambda_S$ and $\gamma_D$ enter \cref{ode_model} only through the product $\lambda_S(1-\gamma_D)$, leaving the off-drug rate $\lambda_S$ and $\gamma_D$ unidentified. 
    Under no treatment ($T_\mathrm{x}\equiv0$), $\gamma_D$ is absent and unidentified. 
    By \cref{lem:transport}, neither time-invariant source schedule transports to targets that depend on the parameters it cannot identify (\cref{tab:matrix}). 
    In what follows, $O(\cdot)$ describes how big a quantity is relative to the quantity in its argument, in the limit $f_0\to0$ (see the supplementary material for details).

	\medskip
    \noindent\textit{Time-variant source schedules.}
    A time-variant schedule visits both drug-on and drug-off states, so it identifies $\lambda_S$ and $\gamma_D$ separately, unlike a time-invariant one. 
    Whether it also identifies the resistant growth rate $\lambda_R$ depends on how large the resistant population becomes, which is not known in advance (see e.g., \cref{fig2}, which shows that the inferred resistant population in Patient 26 remains small). 
    An intermittent schedule with sufficiently long holidays releases resistance (the resistant population reaches $O(1)$), so $\lambda_R$ is potentially well constrained and the source transports broadly (\cref{tab:matrix}).
    Adaptive therapy, by contrast, is designed to keep resistance suppressed, ideally so that $r$ stays of order $f_0$ throughout ($r=O(f_0)$), leaving the tumour burden $n=s+r$ sensitive-dominated. 
    he data then identify the initial fraction $f_0$ (through $s_0=(1-f_0)n_0$) but only weakly constrain $\lambda_R$, whose influence on the burden is negligible (\cref{prop:weak_lambdaR}).

	\begin{proposition}[Weak identifiability of resistant growth under adaptive therapy]
    \label{prop:weak_lambdaR}
    Let $F_{\mathrm{AT}}$ and $F_{\mathrm{cont}}$ denote the Fisher information $F_\mathcal{S}$ under adaptive and continuous therapy, respectively, and $e_{\lambda_R}$ the unit vector in the $\lambda_R$ direction.
	If the resistant fraction stays small, $O(f_0)$, throughout an adaptive therapy observation window, then the $\lambda_R$ diagonal entry $(F_{\mathrm{AT}})_{\lambda_R\lambda_R}=e_{\lambda_R}^\intercal F_{\mathrm{AT}}\,e_{\lambda_R}=O(f_0^2)$ whereas $(F_{\mathrm{cont}})_{\lambda_R\lambda_R}=O(1)$,
	so the transportability index \cref{eq:rho} satisfies 
    $\rho_{\mathrm{AT}\to\mathrm{cont}}(e_{\lambda_R})=O(f_0^{-2})\to\infty$ as $f_0\to0$.
	\end{proposition}

	\noindent A proof is available in the supplementary material.
	Continuous and intermittent targets depend on the resistant growth rate $\lambda_R$, along which $\rho_{\mathrm{AT}\to\schnew}\gg1$ (\cref{prop:weak_lambdaR}), and are therefore non-transportable from adaptive source data.
	This is most severe precisely when adaptive therapy works best: the smaller the resistant fraction $f_0$, the more uncertainty in estimates of $\lambda_R$. 
	Since $\lambda_R$ also governs whether relapse occurs (\cref{sec:relapse}), its weak identifiability limits prediction of the clinical outcomes of interest.

	\begin{table}[h!]
	\centering
	\begin{tabular}{l cccc}
	\toprule
	& \multicolumn{4}{c}{\textbf{target}}\\
	\cmidrule(lr){2-5}
	\textbf{source} & no drug & adaptive & intermittent & continuous \\
	\midrule
	no drug                & $\checkmark$ & $\times$       & $\times$       &$\times$ \\
	adaptive               & $\checkmark$ & $\checkmark$   & $\times^{\ast}$ & $\times^{\ast}$\\
	intermittent$^{\dagger}$ & $\checkmark$ & $\checkmark$   & $\checkmark$   & $\checkmark$\\
	continuous             & $\times$     & $\times$       & $\times$       & $\checkmark$ \\
	\bottomrule
	\end{tabular}
	\caption{
	\textbf{The transportability matrix.}
	Each entry indicates whether parameters calibrated on an observed source schedule can be transported to predict a target schedule, under the model in \cref{ode_model}.
	Markers indicate transportable ($\checkmark$), structurally non-transportable ($\times$), and practically non-transportable ($\times^{\ast}$) schedules.
    Diagonal entries are trivially transportable. 
	Structural non-transportability is exact even with noise-free data, whereas practical non-transportability arises only under measurement noise (\cref{prop:weak_lambdaR}).
	$^{\dagger}$The intermittent row assumes a source that spans both low- and high-resistance regimes; a resistance-suppressing intermittent source behaves like the adaptive row.
	}
	\label{tab:matrix}
	\end{table}

	\subsubsection{Misspecification}
	\label{sec:transport-mis}
	In our original model \cref{ode_model}, we assumed that the resistant population is not directly affected by treatment. Here we instead consider a ground-truth scenario in which the proliferation rate of resistant cells is reduced by 50\% under treatment, so that
    \begin{equation}\label{ode_model2}
	\begin{aligned}
		\dv{s(t)}{t} &= \lambda_S s(t) \left(1 - \dfrac{n(t)}{K}\right) (1 - \gamma_D T_\mathrm{x}(t)) - \gamma_T s(t),\\
		\dv{r(t)}{t} &= \lambda_R r(t) \left(1 - \dfrac{n(t)}{K}\right) (1 - 0.5 T_\mathrm{x}(t))- \gamma_T r(t),
	\end{aligned}
	\end{equation}
    where we take $T_\mathrm{x}(t) = 1$ if a patient is on treatment, and $T_\mathrm{x}(t) = 0$ if a patient is off treatment. In this section, we treat \cref{ode_model2} as the true data-generating process and investigate the ability of the assumed model \cref{ode_model} both to reproduce an observed adaptive schedule and to make cross-schedule predictions. The two models differ only in how they treat the resistant proliferation rate: the assumed model uses the same rate on and off treatment, whereas the true process reduces it during treatment. 
    The best single-rate fit to a source schedule is then neither $\lambda_R^{\mathrm{on}} := 0.5\,\lambda_R$ nor $\lambda_R^{\mathrm{off}} := \lambda_R$, but an effective rate $\bar\lambda_R$ that reproduces the  resistant growth over the observation window.
    To a good approximation (when $1-n/K$ varies little over the observation window), $\bar\lambda_R$ is the average of the on- and off-treatment rates weighted by the on-treatment fraction $\varphi=T_{\mathrm{on}}/T$, 
	\begin{equation}\label{eq:effrate}
	\bar\lambda_R
	\approx\varphi\,\lambda_R^{\mathrm{on}}+(1-\varphi)\,\lambda_R^{\mathrm{off}}. 
	\end{equation}
	Since $\varphi$ differs between schedules, the effective rate $\bar\lambda_R$ is schedule-specific. 
	Inferring model parameters from data generated under a source schedule $\schobs$ (with on-treatment fraction $\varphi_\schobs$) yields the effective rate $\bar\lambda_R^{(\schobs)}$. 
    Accurately predicting the dynamics under a different target schedule $\schnew$ (with its own on-treatment fraction $\varphi_\schnew\neq\varphi_\schobs$) would require $\bar\lambda_R^{(\schnew)}$, which inference on the source does not provide. 
	Transporting a source fit from $\schobs$ to a target $\schnew$ thus incurs the misspecification transport bias $\mu_{\schobs \rightarrow \schnew}$,
	\begin{equation}\label{eq:misspec_error}
	\mu_{\schobs \to \schnew}=\bar\lambda_R^{(\schobs)}-\bar\lambda_R^{(\schnew)}\approx(\varphi_\schobs-\varphi_\schnew)(\lambda_R^{\mathrm{on}}-\lambda_R^{\mathrm{off}}),
	\end{equation}
	which we expect to be non-zero whenever the schedules differ in on-treatment fraction ($\varphi_\schobs\neq\varphi_\schnew$) and the growth rate of the resistant population is genuinely drug-dependent ($\lambda_R^{\mathrm{on}}\neq\lambda_R^{\mathrm{off}}$).
	No choice of source schedule removes this bias: the model itself is misspecified.

\enlargethispage{\baselineskip}
    To illustrate these results, we generated synthetic data from the ground-truth dual-rate model (\cref{ode_model2}), using the MAP estimate from Patient 99 (considered earlier) under an AT-50 schedule on a 14-day decision interval. The assumed, misspecified single-rate model (\cref{ode_model}) reproduces the dynamics under adaptive therapy (\cref{fig9}a). However, \cref{fig9}b,c show that cross-schedule predictions are unreliable.

	\begin{figure}[!t]
		\centering
		\includegraphics{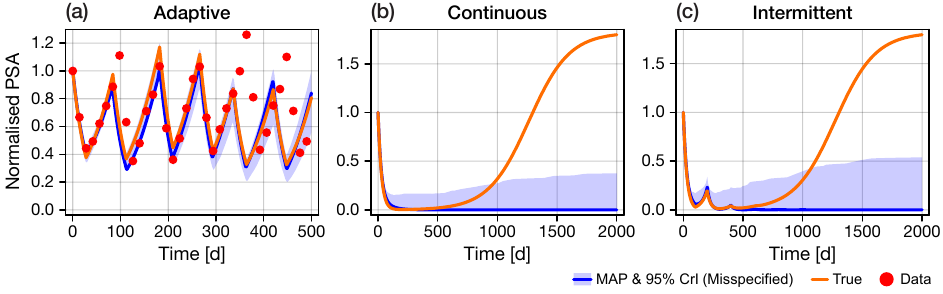}
		\caption[Figure 9]{\textbf{Cross-schedule predictions under misspecification.} 
        Synthetic data are generated from the true dual-rate model (\cref{ode_model2}) under an AT-50 schedule (14-day decision interval), using the Patient~99 MAP estimate. The misspecified single-rate model (\cref{ode_model}) is then calibrated to these data.
        (a) Data with the calibrated single-rate fit (MAP and 95\% credible interval). (b,c) Single-rate posterior predictions under (b) continuous therapy and (c) an intermittent schedule (100 days on, 100 days off), each compared with the dual-rate ground truth.
        }
		\label{fig9}
	\end{figure}

\section{Discussion and conclusions}

Adaptive therapy has emerged as a promising alternative to treatments based on the continuous application of a maximum tolerated dose. 
In practice, however, inter-patient heterogeneity and potentially significant measurement uncertainty can impact the application and eventual outcome of any adaptive treatment regime. In this work, we incorporate these sources of uncertainty into a mathematical and statistical analysis of adaptive therapy. By calibrating a simple competition model to longitudinal PSA measurements, we construct a virtual cohort of patients that captures intra-patient heterogeneity, within-patient uncertainty, and measured biomarker uncertainty. Within our framework, we find that adaptive therapy can still provide a TTP benefit compared to continuous treatment. Simultaneously, our results highlight several limitations in the application of deterministic mathematical models to predict patient outcomes.

A high likelihood of adaptive therapy providing prolonged TTP follows directly from the ecological assumptions underlying the mathematical model describing the tumour dynamics. Adaptive therapy exploits competition between sensitive and resistant subpopulations by maintaining a sufficiently large sensitive population to suppress resistant growth. Consequently, in a modelling framework whereby resistance carries a fitness cost and tumour dynamics are governed by competitive interactions, adaptive schedules are expected to delay inevitable progression relative to continuous treatment, with the magnitude of this improvement dependent upon parameter values. Our analysis, however, indicates that progression is not always inevitable: our posterior probability of relapse, determined from the stability of model equilibria, is approximately 43\%. The majority of patients in our virtual cohort, therefore, would likely benefit from continuous treatment which, in these cases, would lower tumour burden and our new metric, TTM, without the risk of progression. In future, our statistical framework could be applied prospectively to predict whether a patient is likely to benefit from adaptive therapy before they have finished their course of treatment \cite{Browning.2024a}.

Our pseudo-hierarchical Bayesian approach yields a virtual cohort of patients that captures both heterogeneity between patients and residual uncertainty in parameter estimates. Both sources of uncertainty are carried across to clinically relevant predictions, which comprise distributions and not point estimates. Our additional incorporation of measurement uncertainty allows us to quantify the variability in outcomes both within individual patients and across our entire cohort. As we demonstrate in \cref{sec:transport}, inferring a full range of possibilities is critical since point estimates may not readily transport between different treatment regimes. Most obviously, data from patients that do not experience any lapse in treatment cannot be used to predict adaptive nor intermittent therapies. Less obvious is that data from intermittent or adaptive therapy will not give informative predictions unless a patient's tumour is observed when both sufficiently resistant and sufficiently sensitive: a requirement that is impossible to verify retrospectively. Even accounting for residual parametric uncertainty through our Bayesian approach, our predictions are limited by our choice of model. As we show in \cref{fig9}, an adequate fit may not translate to reliable predictions. Future work should further consider model misspecification and, potentially, multimodel predictions \cite{Linden-Santangeli.2024}.

All population-level inferences made through our virtual cohort are necessarily biased by the underlying clinical data, including both the original recruitment of patients onto the clinical trial and our exclusion of patients for whom an adequate model fit could not be obtained. In particular, 12 of the 85 patients (14\%) who met the five-measurement threshold yet had PSA dynamics that the statistical model could not capture well, and thus were removed. Our cohort is, therefore, additionally restricted by construction to patients represented within our framework, which risks overstating both its predictive reliability and the apparent benefit of adaptive therapy. Such selection effects echo a central concern in Mistry's critique of the initial adaptive therapy trial \cite{Mistry.2021} that the reported superiority of adaptive over continuous dosing rests on comparisons to cohorts that are prognostically unmatched and, in one comparison, restricted to patients with an initial PSA decline of more than 50\% \cite{Zhang.2017}. Particularly in light of the response from Zhang et al. \cite{Zhang.2021}, our work demonstrates how sources of uncertainty---including the transferability of parameters between treatment regimes, model and parametric uncertainty, and finally measurement uncertainty---should be considered when assessing the efficacy of adaptive therapy.

The focus of our analysis is on an established mathematical model paired with a relatively simple measurement process that captures the imperfect correlation between measured PSA and GTV. Distinguishing our statistical model from others that consider either independent noise or (in some sense, equivalently) a least-squares estimate is that measurements are subject to temporal autocorrelation \cite{Strobl.2021}. Importantly, this means that repeatedly sampling a patient's PSA does not necessarily lead to more accurate measurements of GTV: this is an implicit assumption behind uncorrelated noise models that is inappropriate to describe the relationship between PSA and GTV. This autocorrelation can also be motivated by considering an extended model in which secretion and degradation of PSA are modelled explicitly \cite{Brady-Nicholls.2020}, and allows us to simulate future schedules using only an ODE model, and to perform inference with a relatively small computational cost. 

An assumption underlying 
many deterministic models of cancer is that the underlying dynamics remain fixed. In cancer, where the timescales often stretch into years (recall that the average duration in our dataset is more than four years), this is unlikely to hold true. We demonstrate two stochastic extensions to our model that illustrate how this additional uncertainty could be incorporated without significantly increasing model complexity. First, motivated by observations that appear to show differences in niche size across treatment windows (\cref{fig7}), we consider a carrying capacity that experiences both short-term fluctuations and long-term drift. For a fixed parameter set, it is clear that this simple change can introduce significant, and potentially realistic, variability in a patient's GTV trajectory. Other future extensions should include, for example, the development of multiple independent niches: circulating PSA gives only a measure of \textit{total} GTV, whereas the model assumes competition within only a single niche. The development of multiple lesions through, for example, metastasis would change the underlying dynamics. Additionally, our model does not capture spontaneous mutations to the resistant state \cite{Monro.2009}, an additional risk factor associated with potentially large drug-sensitive tumours.

The second extension we consider is an explicit model of stochasticity in PSA, where the secretion rate (assumed to be identical between sensitive and resistant cells) fluctuates, with these fluctuations forming the dominant source of variability in measured PSA in our model. This relatively straightforward extension leads to a stochastic differential equation model that can be solved semi-analytically to derive a temporally correlated noise model with time-dependent variance. Importantly, this leads to a relatively small additional computational cost, compared with the stochastic carrying capacity model, the latter requiring a computationally-expensive particle filter approach for inference. This model also demonstrates how even precise measurements of PSA are an imperfect measure of GTV, with the peak in PSA occurring after that of GTV: in an adaptive schedule, this would mean that thresholds based on PSA lead to a small delay in treatment decisions compared with hypothetical direct measurements of GTV. We emphasise that our two extensions are not intended to be exhaustive, but rather serve to motivate the development of models that capture additional sources of clinically relevant uncertainty.

Overall, our results demonstrate that uncertainty can have important ramifications for predicted treatment outcomes. While adaptive therapy remains a promising treatment paradigm, reliable patient-specific prediction requires robust uncertainty quantification of model parameters, the model structure, and potentially also the underlying dynamics which are likely to evolve in time. Importantly, outcomes are also multi-dimensional: time-to-progression captures neither the metastasis risk of a sustained tumour burden nor the reduced cost and toxicity that adaptive therapy offers. Weighing these trade-offs under genuine uncertainty is requisite when establishing whether, and for whom, adaptive therapy is truly beneficial.

\section*{Acknowledgements}

We thank Adriana Zanca and Adrianne Jenner for helpful discussions.
SH was supported by Wenner-Gren Stiftelserna/the Wenner-Gren Foundations (WGF2022-0044) and the Kjell och M{\"a}rta Beijer Foundation. RMC is supported by the Engineering and Physical Sciences Research Council (EP/Z534870/1). 

\section*{Data availability}

Code used to produce the results is available at \url{https://github.com/ap-browning/robust-adaptive-therapy}.


\clearpage
\section*{\centering Supplementary Material}

\renewcommand{\theequation}{S\arabic{equation}}
\renewcommand{\thesection}{S\arabic{section}}
\renewcommand{\thefigure}{S\arabic{figure}}
\renewcommand{\thetable}{S\arabic{table}}
\setcounter{section}{0}
\setcounter{equation}{0}

\section{Equilibria of the two-species Lotka-Volterra model}

	We are interested in the stability of our model in the case that $T_\mathrm{x}(t) = 1$ (e.g., the treatment is on). In this case, the dynamics are given by
	\begin{equation*}\label{ode_model_Txon}
	\begin{aligned}
		\dv{s(t)}{t} &= \lambda_S s(t) \left(1 - \dfrac{n(t)}{K}\right) (1 - \gamma_D) - \gamma_T s(t),\\
		\dv{r(t)}{t} &= \lambda_R r(t) \left(1 - \dfrac{n(t)}{K}\right) - \gamma_T r(t).
	\end{aligned}
	\end{equation*}
	The system has three equilibria, given by
	\begin{align*}
		(s_1^*,r_1^*) &= (0,0),\\
		(s_2^*,r_2^*) &= K\left(1 - \dfrac{\gamma_T}{\lambda_S(1 - \gamma_D)},0\right),\\
		(s_3^*,r_3^*) &= K\left(0,1 - \dfrac{\gamma_T}{\lambda_R}\right).
	\end{align*}
	R3elapse will only occur when the stable equilibrium supports a non-zero population of resistant cells. As such, at $(s_3^*,r_3^*)$, the Jacobian is given by
		\begin{equation}
			\begin{pmatrix}
				-\dfrac{\gamma_T(\lambda_R - \lambda_S(1 - \gamma_D)}{\lambda_R} & 0 \\
				\gamma_T - \lambda_R & \gamma_T - \lambda_R
			\end{pmatrix},
		\end{equation}
	which clearly has eigenvalues given by
		\begin{align*}
			\lambda_1 &= -\dfrac{\gamma_T(\lambda_R - \lambda_S(1 - \gamma_D))}{\lambda_R},\\
			\lambda_2 &= \gamma_T - \lambda_R.
		\end{align*}
	As all parameters are rates and are, therefore, positive, relapse will eventually occur provided
		\begin{equation}
			\lambda_R > \max\left( \gamma_T, \lambda_S(1 - \gamma_D) \right).
		\end{equation}
	That is, relapse will occur if the resistant population is self sustaining and if the intrinsic resistant cell growth rate is greater than the growth rate of sensitive cells on drug.

\clearpage
\section{Derivation of stochastic PSA secretion likelihood}

	Here, we demonstrate that
		\begin{equation*}\label{stoch_psa_model_cov}
			\mathrm{cov}(p(t_1),p(t_2)) = \dfrac{\eta_\alpha}{2\xi_\alpha} \Big(A_+(t_1)A_-(t_2) + B_+(t_1) - B_-(t_1)\Big),
		\end{equation*}
	where
		\begin{equation}
			\dv{A_\pm(t)}{t} = n(t) \mathrm{e}^{\pm \xi_\alpha t}\qquad\text{and}\qquad\dv{B_\pm}{t} = A_\pm(t) n(t) \mathrm{e}^{\mp \xi_\alpha t},
		\end{equation}
	subject to $A_\pm(0) = B_\pm(0) = 0$.

	First, it is helpful to note that
	\begin{equation}
		A_\pm(t) = \int_0^t n(s) \mathrm{e}^{\pm \xi_\alpha s}\,\dd s\qquad \text{and}\qquad 
		B_\pm(t) = \int_0^t A_\pm(s) n(s) \mathrm{e}^{\mp \xi_\alpha s}\,\dd s.
	\end{equation}
	Now, consider that
		\begin{align*}
			\dfrac{2\xi_\alpha}{\eta_\alpha} \mathrm{cov}(p(t_1),p(t_2)) &= \int_0^{t_1}\int_0^{u_1} n(u_1)n(u_2) \mathrm{e}^{-\xi_\alpha(u_1 - u_2)} \dd u_2 \dd u_1 \\&\qquad\quad+ \int_0^{t_1}\int_{u_1}^{t_2} n(u_1)n(u_2) \mathrm{e}^{-\xi_\alpha(u_2 - u_1)} \dd u_2 \dd u_1,\\
			&= \int_0^{t_1}n(u_1)\mathrm{e}^{-\xi_\alpha u_1}\int_0^{u_1} n(u_2) \mathrm{e}^{\xi_\alpha u_2} \dd u_2 \dd u_1\\&\qquad\quad+ \int_0^{t_1}n(u_1) \mathrm{e}^{\xi_\alpha u_1}\int_{u_1}^{t_2} n(u_2) \mathrm{e}^{-\xi_\alpha u_2} \dd u_2 \dd u_1,\\
			&= \int_0^{t_1}n(u_1)\mathrm{e}^{-\xi_\alpha u_1} A_+(u_1)\, \dd u_1 + \int_0^{t_1}n(u_1) \mathrm{e}^{\xi_\alpha u_1} \left(A_-(t_2) - A_-(u_1) \right) \dd u_1,\\
			&= B_+(t_1) + A_-(t_2)  \int_0^{t_1}n(u_1) \mathrm{e}^{\xi_\alpha u_1} \dd u_1 - \int_0^{t_1}n(u_1) \mathrm{e}^{\xi_\alpha u_1} A_-(u_1) \dd u_1,\\
			&= B_+(t_1) + A_+(t_1) A_-(t_2) - B_-(t_1),
		\end{align*}
	which gives the desired result.

\clearpage
\section{Details on the transportability condition (Lemma 1)}
\paragraph{Sensitivity matrix and nullspace.}
The sensitivity matrix $J_\schobs$ is the $N\times p$ matrix ($N$ is the number of time points available for the patient, and $p=\dim\bm\theta$) with entries
\begin{equation}
(J_\schobs)_{ij} = \left.\frac{\partial G_\schobs(\bm\theta;t_i)}{\partial\theta_j}\right|_{\bm\theta^\ast}
= \left.\frac{\partial n(t_i)}{\partial\theta_j}\right|_{\bm\theta^\ast}.
\end{equation}
Here $\theta_j\in\{\lambda_S,\lambda_R,\gamma_D,\gamma_T,f_0\}$, so that $p=5$.
Each column is obtained by solving the forward-sensitivity equations, found by differentiating \cref{ode_model} with respect to $\theta_j$.
The nullspace of $J_\schobs$, denoted by $\mathcal{N}_\schobs=\ker J_\schobs$, collects the parameter perturbations $\delta\bm\theta$ satisfying $J_\schobs\,\delta\bm\theta=0$. These are the parameter directions that leave the observable from the source schedule unchanged to first order, i.e., the locally structurally non-identifiable directions used in \cref{lem:transport}.

\paragraph{Lemma 1.}
\begin{proof}
First, consider the noise-free source data $G_\schobs(\bm\theta^\ast;t_i)$, $i=1,\dots,N$.  
A model parameter set $\bm\theta^\ast+\delta\bm\theta$ reproduces the data values iff $G_\schobs(\bm\theta^\ast+\delta\bm\theta;t_i)=G_\schobs(\bm\theta^\ast;t_i)$ $\forall i$. 
A first-order Taylor expansion of $G_\schobs$ about $\bm\theta^\ast$ gives
\begin{equation}
G_\schobs(\bm\theta^\ast+\delta\bm\theta;t_i)=G_\schobs(\bm\theta^\ast;t_i)+\big[J_\schobs\,\delta\bm\theta\big]_i+O(\|\delta\bm\theta\|^2),
\end{equation}
thus the data are reproduced to leading order iff $J_\schobs\,\delta\bm\theta=0$, i.e.,\ $\delta\bm\theta\in\ker J_\schobs=\mathcal{N}_\schobs$. 
The source data therefore determine $\bm\theta$ only up to an additive element of $\mathcal{N}_\schobs$: they cannot distinguish $\bm\theta^\ast$ from $\bm\theta^\ast+\delta\bm\theta$ for any $\delta\bm\theta\in\mathcal{N}_\schobs$, pinning $\bm\theta$ down only to the set $\bm\theta^\ast+\mathcal{N}_\schobs$.
\medskip

\noindent Second, a first-order Taylor expansion of a scalar target quantity $g_\schnew$ about $\bm\theta^\ast$ gives
\begin{equation}
g_\schnew(\bm\theta^\ast+\delta\bm\theta)=g_\schnew(\bm\theta^\ast)+\nabla g_\schnew(\bm\theta^\ast)^\intercal\delta\bm\theta+O(\|\delta\bm\theta\|^2).
\end{equation}
Across all data-consistent parameters (i.e.,\ those with $\delta\bm\theta\in\mathcal{N}_\schobs$), 
$g_\schnew$ takes the same value to leading order iff 
$\nabla g_\schnew(\bm\theta^\ast)^\intercal\delta\bm\theta=0$ 
for every $\delta\bm\theta\in\mathcal{N}_\schobs$. 
This holds precisely when $\nabla g_\schnew(\bm\theta^\ast)\perp\mathcal{N}_\schobs$. 
Otherwise $g_\schnew$ varies across the data-consistent set: 
a fit displaced from the truth by a $\delta\bm\theta\in\mathcal{N}_\schobs$ yields a value differing from the true $g_\schnew(\bm\theta^\ast)$ by $\nabla g_\schnew(\bm\theta^\ast)^\intercal\delta\bm\theta$, which is the transport error. 
\end{proof}

\clearpage
\section{Weak identifiability of resistant growth (Proposition 1)}

\medskip
\noindent Throughout, we let $O(\cdot)$ describe how big a quantity is relative to the quantity in its argument: $X=O(Y)$ means $X$ is of size $Y$. For example, as $f_0\to0$, $O(1)$ stays of order one, $O(f_0)$ vanishes, and $O(f_0^{-2})$ diverges; and $X=O(r(t))$ means $X$ is of the size of the resistant population $r(t)$. The observation window $[0,t]$ is fixed and finite; the only limit we take is $f_0\to0$.
\\

\noindent{\bf Proposition 1.}

\begin{proof}
Integrating the equations for $r(t)$ and $s(t)$ in \cref{ode_model} gives
\begin{equation}
r(t)=r_0\,e^{E(t)}, \qquad s(t)=s_0\,e^{D(t)},
\label{eq:supmat_sols}
\end{equation}
where
\begin{align}
E(t)&=\int_0^t\big[\lambda_R\big(1-n(\tau)/K\big)-\gamma_T\big]\,\dd\tau,\\
D(t)&=\int_0^t\big[\lambda_S\big(1-n(\tau)/K\big)\big(1-\gamma_D T_\mathrm{x}\big)-\gamma_T\big]\,\dd\tau,
\end{align}
and $n(t)=s(t)+r(t)$. The parameter $\lambda_R$ appears explicitly in $E(t)$ but not in $D(t)$. It drives the resistant population directly, and the sensitive population indirectly through the shared burden $n(t)$.

Differentiating \cref{eq:supmat_sols} with respect to $\lambda_R$ by the chain rule, the resistant sensitivity is $O(r)$ since
\begin{equation}
\frac{\partial r(t)}{\partial\lambda_R}
= r_0\,e^{E(t)}\,\frac{\partial E(t)}{\partial\lambda_R}
= \underbrace{r(t)}_{O(r)}\,\underbrace{\frac{\partial E(t)}{\partial\lambda_R}}_{O(1)}
= O(r),
\end{equation}
where
\begin{equation}
\frac{\partial E(t)}{\partial\lambda_R}
= \int_0^t\!\Big[\big(1-\tfrac{n(\tau)}{K}\big) - \frac{\lambda_R}{K}\frac{\partial n(\tau)}{\partial\lambda_R}\Big]\dd\tau
\end{equation}
has a bounded integrand ($0\le n\le K$, and $\partial n/\partial\lambda_R$ is bounded on $[0,t]$ because a linear ODE with bounded coefficients on a finite interval has a bounded solution) and is hence $O(1)$.

The sensitivity of the sensitive population has a bounded prefactor, $s(t)=O(1)$ (since $s\le K=1$),
\begin{equation}
\frac{\partial s(t)}{\partial\lambda_R}
= \underbrace{s(t)}_{O(1)}\,\frac{\partial D(t)}{\partial\lambda_R},
\label{eq:ds}
\end{equation}
so any $\lambda_R$-dependence of $s(t)$ enters through $\partial D(t)/\partial\lambda_R$, which is coupled to $\partial n(t)/\partial\lambda_R$ by
\begin{equation}
\frac{\partial D(t)}{\partial\lambda_R}
= -\frac{\lambda_S}{K}\int_0^t\big(1-\gamma_D T_\mathrm{x}\big)\,\frac{\partial n(\tau)}{\partial\lambda_R}\,\dd\tau .
\label{eq:dD}
\end{equation}
We therefore bound the total sensitivity $\partial n(t)/\partial\lambda_R = \partial s(t)/\partial\lambda_R + \partial r(t)/\partial\lambda_R$, using the coupling \cref{eq:dD}, rather than $\partial s(t)/\partial\lambda_R$ alone.
Substituting the coupling \cref{eq:dD} into \cref{eq:ds}, and using $\partial r(t)/\partial\lambda_R = O(r)$, the total sensitivity becomes
\begin{equation}
\frac{\partial n(t)}{\partial\lambda_R}
= \underbrace{\frac{\partial r(t)}{\partial\lambda_R}}_{O(r)}
- s(t)\,\frac{\lambda_S}{K}\int_0^t\big(1-\gamma_D T_\mathrm{x}\big)\,\frac{\partial n(\tau)}{\partial\lambda_R}\,\dd\tau .
\label{eq:Dn_DlambdaR}
\end{equation}
To bound the LHS in \cref{eq:Dn_DlambdaR}, we use Grönwall's inequality, which states that if a non-negative, continuous function $u(t)$ satisfies
\begin{equation}
u(t) \le \alpha + C\int_0^t u(\tau)\,\dd\tau
\end{equation}
for constants $\alpha,C\ge 0$, then $u(t) \le \alpha\,e^{Ct}$.
Taking absolute values in \cref{eq:Dn_DlambdaR}, and using the triangle inequality, $|a-b|\le|a|+|b|$, together with $\big|\int f\big|\le\int|f|$,
\begin{equation}
\begin{split}
\Big|\frac{\partial n(t)}{\partial\lambda_R}\Big|
&= \Big|\frac{\partial r(t)}{\partial\lambda_R} - s(t)\,\frac{\lambda_S}{K}\int_0^t\big(1-\gamma_D T_\mathrm{x}\big)\,\frac{\partial n(\tau)}{\partial\lambda_R}\,\dd\tau\Big| \\
&\le \Big|\frac{\partial r(t)}{\partial\lambda_R}\Big| + s(t)\,\frac{\lambda_S}{K}\int_0^t\big|1-\gamma_D T_\mathrm{x}\big|\,\Big|\frac{\partial n(\tau)}{\partial\lambda_R}\Big|\,\dd\tau \\
&\le \alpha + C\int_0^t\Big|\frac{\partial n(\tau)}{\partial\lambda_R}\Big|\,\dd\tau,
\end{split}
\end{equation}
where $\alpha=O(r)$ bounds $|\partial r(t)/\partial\lambda_R|$ and $C$ bounds the kernel $s(t)\lambda_S(1-\gamma_D T_\mathrm{x})/K$. This has the form required by Gr\"onwall's inequality, with $u(t)=|\partial n(t)/\partial\lambda_R|$, so
\begin{equation}
\Big|\frac{\partial n(t)}{\partial\lambda_R}\Big| \le \alpha\,e^{Ct} = O(r),
\label{eq:supmat_dn_dlambdaR}
\end{equation}
since $e^{Ct}$ is constant on the fixed window $[0,t]$.

By \cref{eq:supmat_dn_dlambdaR}, the total tumour burden's sensitivity to $\lambda_R$ scales with $r$. 
Under effective adaptive therapy, which keeps the resistant population suppressed, $r=O(f_0)$, so $\partial n(t_i)/\partial\lambda_R=O(f_0)$ and
\begin{equation}
(F_{\mathrm{AT}})_{\lambda_R\lambda_R}
= \sum_i J_{i,\lambda_R}\,J_{i,\lambda_R}
= \sum_i \Big(\frac{\partial n(t_i)}{\partial\lambda_R}\Big)^2
= \sum_i O(f_0)\cdot O(f_0)
= O(f_0^2).
\label{eq:F_AT}
\end{equation}
Under continuous therapy, however, resistance is released so that $r=O(1)$, giving
\begin{equation}
(F_{\mathrm{cont}})_{\lambda_R\lambda_R}
= \sum_i O(1)\cdot O(1)
= O(1).
\label{eq:F_cont}
\end{equation}
Combining \cref{eq:F_AT,eq:F_cont}, the transportability index satisfies
\begin{equation}
\rho_{\mathrm{AT}\to\mathrm{cont}}(e_{\lambda_R})
= \frac{(F_{\mathrm{cont}})_{\lambda_R\lambda_R}}{(F_{\mathrm{AT}})_{\lambda_R\lambda_R}}
= \frac{O(1)}{O(f_0^2)} = O(f_0^{-2}) \longrightarrow \infty \quad\text{as } f_0\to0. \qedhere
\end{equation}
\end{proof}

\end{document}